\documentclass[12pt,draftcls,onecolumn]{IEEEtran}

\usepackage{setspace,multirow,tikz,etoolbox}
\usetikzlibrary{arrows,automata}
\usepackage{lipsum, enumitem}

\usepackage{blkarray}
\usepackage[T1]{fontenc}
\usepackage[latin9]{inputenc}
\usepackage{amsthm}
\usepackage{amsmath}
\usepackage{graphicx}
\usepackage{amssymb}
\usepackage{url}
\usepackage{verbatim}
\usepackage[export]{adjustbox}
\usepackage{framed}
\usepackage{multirow}
\usepackage[font=footnotesize,labelfont=bf]{caption}

\newtoggle{arxiv}
\toggletrue{arxiv}
\newtoggle{doublecol}
\togglefalse{doublecol}

\newcommand{\arxiv}[1]{\iftoggle{arxiv}{#1}{\cite[#1]{Proofs}}}

\newcommand{\w}[2]{\iftoggle{doublecol}{#1}{#2}}

\newcommand{\vone}[1]{}
\newcommand{\vdo}[1]{\iftoggle{doublecol}{\vspace{#1}}{}}

\theoremstyle{plain}
\newtheorem{thm}{Theorem}
\theoremstyle{plain}

\usepackage{cite}\usepackage{amsfonts}\usepackage{times}\usepackage{bm}
\usepackage{amsthm}
\usepackage{subfigure}
\theoremstyle{definition}

\newtheorem{definition}{Definition}

\newtheorem{lemma}{Lemma}

\newtheorem*{lemma*}{Lemma}

\theoremstyle{remark}
\newtheorem{remark}{Remark}

\newcommand\K{\mathrm{K}}
\newcommand\U{\mathrm{U}}

\title{Optimal Spectrum Sharing with ARQ based \\ Legacy Users via Chain Decoding}

\author{
\textbf{Nicol\`o Michelusi\IEEEauthorrefmark{2}}, \emph{Senior Member, IEEE}
\\
{\small \IEEEauthorrefmark{2}School of Electrical and Computer Engineering,
Purdue University, USA,\\ \texttt{\small michelus@purdue.edu}}
\vdo{-10mm}
\vone{-15mm}
\thanks{This research has been funded in part by the grants 
NSF CNS-1642982 and DARPA \#108818. Part of this work appeared at ISIT'17, see \cite{MicheISIST17}.}
}

\graphicspath{
     {./Fig.s/}
}

\begin{document}
\maketitle
\begin{abstract}
This paper investigates the design of access policies in spectrum sharing networks by exploiting the retransmission protocol of legacy primary users (PUs) to improve the spectral efficiency via opportunistic retransmissions at secondary users (SUs) and \emph{chain decoding}. The optimal policy maximizing the SU throughput under an interference constraint to the PU and its performance are found in closed form. It is shown that the optimal policy randomizes among three modes: \emph{Idle}, the SU \emph{remains idle} over the retransmission window of the PU, to avoid causing interference; \emph{Interference cancellation}, the SU transmits \emph{only after decoding the PU packet}, to improve its own throughput via interference cancellation; \emph{Always transmit}, the SU transmits over the retransmission window of the PU to maximize the future potential of interference cancellation via chain decoding. This structure is exploited to design a stochastic optimization algorithm to facilitate learning and adaptation when the model parameters are unknown or vary over time, based on ARQ feedback from the PU and CSI measurements at the SU receiver. It is shown numerically that, for  a 10\% interference constraint, the optimal access policy yields 15\% improvement over a state-of-the-art scheme without SU retransmissions, and up to $2\times$ gain over a scheme using a non-adaptive access policy instead of the optimal one.
\vone{-4mm}
\end{abstract}

\vdo{-5mm}
\section{Introduction}
The recent proliferation of mobile devices has been exponential in number as well as heterogeneity, leading to spectrum crunch. The tremendous increase in demand of wireless services requires a shift in network design from \emph{exclusive spectrum reservation} to
\emph{spectrum sharing} to improve spectrum utilization \cite{pcast}. Cognitive radios  \cite{DySpAN} enable the coexistence of incumbent legacy users (\emph{primary users}, PUs) and opportunistic users (\emph{secondary users}, SUs) capable of autonomous reconfiguration by learning and adapting to the communication environment \cite{MicheTCNC}. 

A central question is: how can opportunistic users  leverage side information about 
 nearby legacy users (\emph{e.g.}, activity, channel conditions, protocols employed, packets exchanged \cite{Geirhofer}) to opportunistically access the spectrum and improve their own performance, with minimal or no degradation to existing legacy users~\cite{Peha}?
In this paper, we address this question in the context of the retransmission protocol
employed by PUs. We consider a wireless network composed of a pair of PUs and a pair of SUs. The PU employs Type{-}I HARQ~\cite{Comroe} to improve reliability, which
 results in replicas of the  PU packet (re)transmitted over subsequent slots, henceforth referred to as \emph{ARQ window}. With the scheme developed in \cite{IT_ARQ}, the SU receiver attempts to decode the PU packet independently in each slot, and replicas of the PU packet are not exploited; thus, in the example of Fig.~\ref{figexlabel2}, no SU packets can be decoded with the scheme \cite{IT_ARQ}. However, the SU may leverage these replicas to improve its own throughput via \emph{interference cancellation}. In \cite{MichelusiJSAC}, we have investigated a scheme, termed \emph{backward interference cancellation} (BIC),  where the SU receiver decodes the PU packet and removes its interference to achieve interference-free transmissions over the entire ARQ window of that PU packet. In the example of Fig.~\ref{figexlabel2}, this scheme allows the SU receiver to decode packet S4, after removing the interference of P2, decoded in slot 5, thus outperforming \cite{IT_ARQ}.
 
  In \cite{MichelusiCD}, we have advanced this concept by allowing the SU to opportunistically retransmit SU packets and buffer the corrupted signals at the SU receiver.
In fact, if a previously transmitted and failed SU packet is decoded at the SU receiver,
its interference can be removed from previous retransmission attempts of the same,
thus facilitating the decoding of the concurrent PU packets; in turn, the interference of these PU packets can be removed to facilitate the decoding of SU packets over their respective ARQ windows. This scheme continues in chain, until no more packets can be decoded, hence the name \emph{chain decoding} (CD) \cite{MichelusiCD}. In the example of Fig.~\ref{figexlabel2}, the retransmission of S3 in slot 4 allows the SU receiver to connect in \emph{chain} the ARQ windows of P1 and P2, so that all 3 SU packets S1-S3 can be decoded, versus only one decoded with BIC, and none decoded with the scheme in  \cite{IT_ARQ}. However, note that the decoding of S1 is delayed by 4 slots. Therefore, the throughput improvement of chain decoding comes at a latency cost in the delivery of SU packets, hence it is suitable for latency-tolerant applications, such as monitoring sensor networks as in \cite{Nayak} and video streaming, see \cite[Fig. 1]{GSMA} for a list of potential use cases. Additionally, as explained in Sec.~\ref{buffering}, chain decoding requires a buffering mechanism at the SU receiver, whereas no buffering is required in \cite{IT_ARQ}. The impact of these factors are evaluated numerically in Sec.~\ref{sec:numres}.
  
  \begin{figure}
    \centering
\includegraphics[width=\w{.7}{.44}\linewidth]{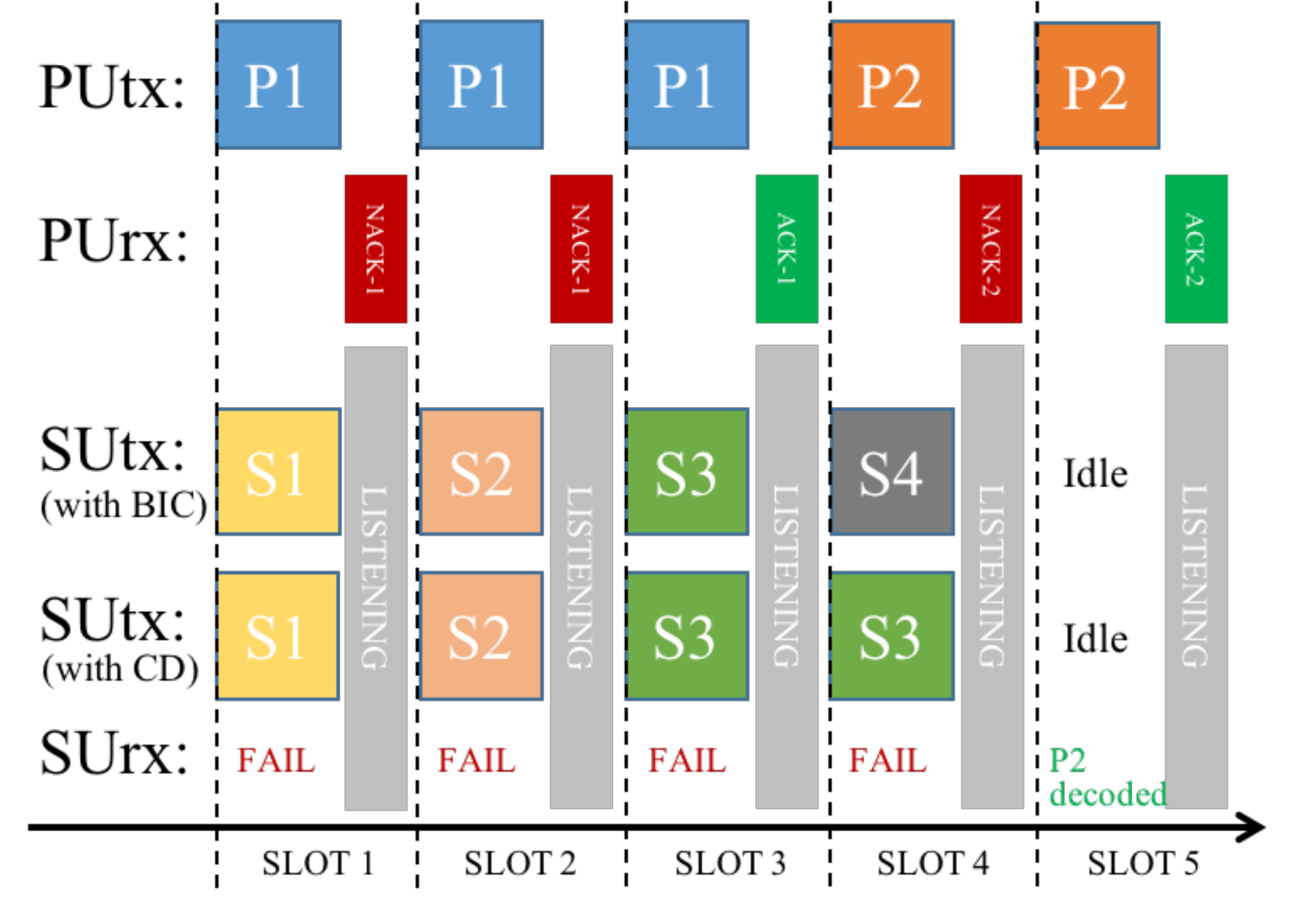}
\caption{Example of chain decoding and comparison with BIC \cite{MichelusiJSAC} and the scheme in \cite{IT_ARQ}. The SU fails in slots 1-4. SUrx decodes P2 in slot 5. With the scheme in \cite{IT_ARQ}, SUrx does not remove the interference of P2 in slot 4, hence no SU packets are decoded. With BIC, SUrx removes the interference of P2 from slot 4 to decode S4, thus decoding $1$ SU packet. With chain decoding, SUrx retransmits S3 in slot 4; after decoding P2 in slot 5, \emph{chain decoding} is initiated: SUrx removes the interference of P2 from slot 4 to decode S3; hence, it removes the interference of S3 from slot 3 to decode P1; finally, it removes the interference of P1 from slots 1-2 to decode S1-S2; overall, SUrx decodes $3$ SU packets. Chain decoding requires a buffering scheme at SUrx, and to monitor the ACK/NACK feedback $y_{P,t}$ from the PU, along with the feedback from SUrx $y_{S,t}$ to decide when and which packet to transmit.}
\label{figexlabel2}
\vone{-10mm}
\vdo{-5mm}
\end{figure}

While our previous work \cite{MichelusiCD} proves the optimality of a CD protocol, which dictates the retransmission process at the SU to maximize the potential of interference cancellation at its receiver, it does not investigate the design of an \emph{optimal SU access scheme} (\emph{i.e.}, whether the SU should transmit or remain idle). Such design question, not addressed in \cite{MichelusiCD} but investigated in this paper, is of great practical interest. In fact, as we will show in Sec.~\ref{learning}, information on the structure of the optimal SU access scheme may be exploited to significantly reduce the policy search space and the optimization complexity, thus facilitating learning and adaptation in scenarios where the statistics are unknown or vary over time.
\vdo{-5mm}
\vone{-5mm}
\subsection{Main Contributions}
Based on the underlay paradigm \cite{goldsmith}, in this paper we investigate the optimal SU access policy to maximize the SU throughput via CD, under an interference constraint to the PU. The contributions of this paper are as follows:
\begin{enumerate}[leftmargin=0.45cm]
\item  We derive the optimal policy and its performance in closed form for the case where the PU enforces reliability, and show that the optimal SU access policy reflects a randomization among three modes of operation: 1) The SU remains idle over the  entire ARQ window;
  2) The SU transmits only after its receiver decodes the PU packet; 3) The SU always transmits over the ARQ window. With mode 1), the SU does not interfere with the PU; with mode 2), it leverages knowledge of the PU packet to perform interference cancellation and create an interference-free channel for its own data transmission; with mode 3), it leverages the full potential of successive interference cancellation via CD over the ARQ window. The optimal randomization among these three modes reflects a strike between maximizing the SU throughput and minimizing the interference to the PU.
  \item
 We show numerically that, for a 10\% interference constraint, the optimal access policy under CD attains a throughput gain of 15\% with respect to BIC, and up to $2\times$ improvement over a CD scheme using a non-adaptive access policy. We demonstrate robustness of CD under a finite buffer size at the SU receiver, and under a finite ARQ deadline at the PU.
\item  Importantly, the optimal  policy does not require knowledge of the statistics of the model, but only an estimate of  the interference level perceived at the PU receiver (for instance, estimated by monitoring the ACK/NACK feedback signal \cite{Eswaran}). This feature facilitates learning and adaptation when the statistics of the system are unknown or vary over time. For these scenarios, we present a stochastic optimization framework, where the SU learns the optimal randomization and its transmit rate based solely on ARQ feedback from the PU and CSI measurements at the SU receiver. We prove the effectiveness of this strategy numerically.
\end{enumerate}
\vdo{-3mm}
\vone{-5mm}
\subsection{Related work}  
Other previous work leverage the retransmission protocol of the PU \cite{Nosratinia,SHARP,IT_ARQ,Jovicic,Kulkarni,Eswaran}. In \cite{Nosratinia,SHARP}, the primary ARQ process is limited to one retransmission, with incremental redundancy and packet combining, respectively, assuming a slow-fading scenario. In  \cite{IT_ARQ}, it is shown that the SU throughput is maximized by concentrating the interference to the PU in the first transmissions of the PU packet, but the temporal redundancy of ARQ is not exploited to cancel interference. In \cite{Kulkarni}, the SUs cooperate with the PU by assisting retransmissions of failed packets using distributed orthogonal space-time block code; however, knowledge of the PU packet is not exploited at the SU receiver to perform interference cancellation. Differently from \cite{Kulkarni}, we assume no cooperation with the PU at the SU transmitter, but only interference cancellation at the SU receiver. Differently from these works, in this paper we consider multiple retransmissions (in contrast to \cite{Nosratinia,SHARP}), and we exploit the redundancy of the ARQ process (in contrast to \cite{IT_ARQ,Kulkarni}). 

Similarly to \cite{Zhang}, we assume that the ARQ feedback is overheard by the SU without errors; practical aspects related to imperfect sensing can be investigated using tools developed in \cite{Kulkarni,Cabric}. While non-causal knowledge of the PU packet is assumed in \cite{Jovicic}, in our work we model the dynamic acquisition of the PU packet at the SU receiver, which is of more practical interest. In \cite{Eswaran}, the SU exploits
ARQ feedback to estimate the throughput loss of the PU and tune its transmission policy, based on information theoretic results. In \cite{Zhang}, the PU adapts its transmit power in response to interference; the SU uses the feedback from the PU to control its interference. 
In this paper, instead, we leverage the structure of the optimal policy to design a simple but effective learning algorithm based on \emph{stochastic gradient descent} \cite[Chapter 14]{Shalev-Shwartz}, as opposed to approaches based on reinforcement learning \cite{Sutton98a}, which suffer from slow convergence rate due to the need to explore the action and state spaces.

This paper is organized as follows. In Sec.~\ref{sec:sys_model}, we describe the system model; in Sec.~\ref{perfopt}, we introduce the performance metrics and optimization problem.
In Sec.~\ref{analysis}, we provide the analytical results. In Sec.~\ref{learning}, we present the stochastic optimization framework, and in Sec.~\ref{sec:numres} we present numerical results. In Sec.~\ref{sec:conclu}, we provide concluding remarks.
\begin{figure}
    \centering
    \includegraphics[width=\w{1}{.7}\linewidth,trim = 0mm 0mm 30mm 0mm,clip=false]{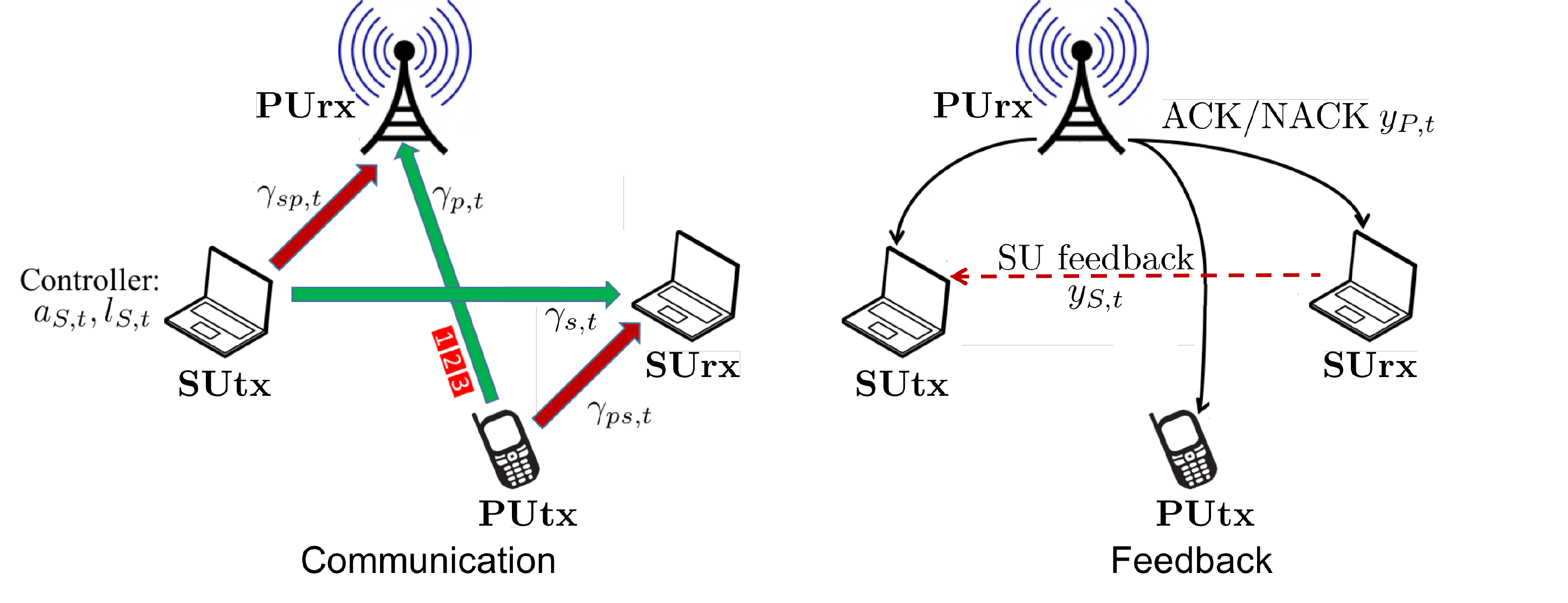}
\caption{System model: the first portion of the slot is devoted to data communication, whereas the last portion is used for feedback signaling.}
\label{figexlabel1}
\vone{-8mm}
\vdo{-5mm}
\end{figure}

\vdo{-3mm}
\vone{-5mm}
\section{System Model}\label{sec:sys_model}
We consider a two-user interference network,  depicted in Fig.~\ref{figexlabel1}, where a primary (PUtx)  and a secondary (SUtx) users transmit to their respective receivers, PUrx and SUrx, and generate mutual interference. The PU uses retransmissions (ARQ) to enforce reliability. The SU, on the other hand, uses chain decoding for its own transmissions. To this end, SUtx monitors the ACK/NACK feedback (signal $y_{P,t}$) from the PU and the state of the buffer at SUrx (via the feedback signal $y_{S,t}$), to decide whether to perform a transmission attempt or remain idle, and which packet to transmit, see example in Fig.~\ref{figexlabel2}. In the following, we provide details of these operations. The main parameters of the model are given in Table I.

Time is divided into slots of fixed duration $\Delta$, corresponding to the transmission of one data packet and the feedback signal from the receiver, see Fig.~\ref{figexlabel2}. We assume a block-fading channel model, \emph{i.e.}, the channel gains are constant within each slot, i.i.d. over time and independent across links. SUtx and PUtx transmit with constant powers $P_s$ and $P_p$, respectively. $P_s$ may be based on an interference temperature threshold experienced at the PU receiver \cite{Zhang}, and can be estimated using
techniques developed in \cite{Nayak}. Assuming AWGN noise at the receivers, we define the SNR of the links SUtx$\rightarrow$SUrx, PUtx$\rightarrow$PUrx, SUtx$\rightarrow$PUrx and PUtx$\rightarrow$SUrx at time $t$ as $\gamma_{s,t},\gamma_{p,t},\gamma_{sp,t},\gamma_{ps,t}$, i.i.d. over time and with mean $\bar\gamma_s,\bar\gamma_p,\bar\gamma_{sp},\bar\gamma_{ps}$, respectively.

No channel state information is available at the transmitters. Thus, PUtx transmits with fixed rate $R_p$ [bits/s/Hz], and is data backlogged. SUtx transmits with fixed rate $R_{s}$ [bits/s/Hz], or remains idle to avoid causing interference to the PU. We denote the SU access decision in slot $t$ as $a_{S,t}\in\{0,1\}$, selected according to access policy $\mu$, introduced in Sec.~\ref{sec:policies}. Thus, $a_{S,t}{=}1$ if SUtx transmits, and
$a_{S,t}{=}0$ if it remains idle.

We assume that the SU knows the signal characteristics of the PU and is accurately synchronized with the PU system, as commonly assumed in the literature \cite{IT_ARQ,Kulkarni,Geirhofer,SHARP,Nosratinia}. The modulation type can be inferred using signal processing techniques such as cyclostationary feature detection \cite{Cabric} or
deep neural networks \cite{ElGamal}. The codebook information may be obtained if the PUs follow a uniform standard for communication based on a publicized codebook, or periodically broadcast it \cite{goldsmith}. Moreover, the SUs perform timing, carrier synchronization and channel equalization by leveraging pilots, preambles, synchronization words or spreading
codes used by PUs for coherent detection \cite{Cabric}. The SU pair can use this information to synchronize with the PU system to detect ARQ feedback messages, decode the PU packet, and then reconstruct the PU transmit signal to perform interference cancellation, 
\emph{e.g.}, using techniques developed in~\cite{Wang}.

 \begin{table}[t]
 \centering
 \footnotesize
\begin{tabular}{ |r|l| } 
 \hline
 \emph{Symbol} & \emph{Meaning} \\ 
  \hline
 $R_s,R_p$ & Transmission rate of SU and PU, bits/s/Hz
 \\
 $P_s,P_p$ & Transmission power of SU and PU
 \\ 
 $l_S,l_P$ & SU and PU packet labels
 \\ 
 $\gamma_{s,t},\gamma_{p,t}$  & 
 SNR of SUtx$\rightarrow$SUrx \& PUtx$\rightarrow$PUrx links \vspace{-0.8mm}\\&  at time $t$,with mean $\bar\gamma_s,\bar\gamma_p$
  \\ 
   $\gamma_{sp,t},\gamma_{ps,t}$ & SNR of SUtx$\rightarrow$PUrx \& PUtx$\rightarrow$SUrx links \vspace{-0.8mm}\\&  at time $t$,  with mean $\bar\gamma_{sp},\bar\gamma_{ps}$
  \\ 
   $a_{S,t}\in\{0,1\}$  & Access decision of SU at time $t$
      \\
     $\mu$ & SU access policy, used to select $a_{S,t}\in\{0,1\}$
  \\ 
  $y_{P,t},y_{S,t}$ &  PUrx and SUrx feedback (see Sec.~\ref{SUFB}), 
  \vspace{-0.8mm}\\&
  $y_{P,t}{\in}\{\text{ACK},\text{NACK}\}$, $y_{S,t}{\in}\{1,\dots,7\}$
  \\
     $\rho_{x},x\in[0,1]$  & Failure probability of PU when SU
       \vspace{-0.8mm}\\& transmits with probability $x$
   \\
   $\delta_s,\delta_p,\delta_{sp},$
   & Decoding probabilities at SUrx,
   \vspace{-0.8mm}\\
   $\upsilon_s,\upsilon_p,\upsilon_{sp},\upsilon_{\emptyset}$&see (\ref{outcomes}) and Fig.~\ref{decevents}\\
   $D_s,D_p$ & 
   Interference-free decoding probability of SU/PU  \vspace{-0.8mm}\\& packets at SUrx, see (\ref{Ds})-(\ref{Dp})
     \\ 
   $C(\mathrm{SNR})$  &  $\triangleq\log_{2}(1+\mathrm{SNR})$, capacity of the Gaussian  \vspace{-0.8mm}\\& channel as a function of SNR
   \\
         $\bar T_S(\mu),\bar T_P(\mu)$ & Average long-term SU and PU throughputs
   \\
   $\bar T_S^{(GA)}(\epsilon)$ & Genie-aided SU throughput, when SUtx  \vspace{-0.8mm}\\& transmits with probability $\epsilon$
   \\
           $\nabla(\mu)$ & Average PU throughput degradation, $\leq\nabla_{\max}$
   \\
           $\nabla_{th}{\in}(0,1)$ & Max throughput degradation tolerated by the PU
   \\
           $
           \{b|b{\geq}0\}{\cup}\{\stackrel{\leftrightarrow}{\K},\stackrel{\rightarrow}{\K}\}
           $ & States of the CD protocol, see Fig.~\ref{CDgraph} and Sec.~\ref{analysis}
   \\
 \hline
\end{tabular}
\label{tablesymbols}
\caption{Main Parameters}
\vone{-10mm}
\vdo{-6mm}
\end{table}

Let $\rho_{a_{S,t}}$ be the failure probability for the PU as a function of $a_{S,t}{\in}\{0,1\}$.
Clearly, $\rho_0{<}\rho_1$, since transmissions of the PU are more likely to fail
under interference from the SU. The PU employs ARQ in case of transmission failure~\cite{Comroe} and enforces perfect reliability, so that a packet is retransmitted until, eventually, it is received successfully. We will evaluate the effect of a finite ARQ deadline numerically in Sec.~\ref{sec:numres}. At the end of slot $t$, PUrx sends a feedback message $y_{P,t}{\in}\{\text{ACK},\text{NACK}\}$ to PUtx over a dedicated control channel, to notify it about the transmission outcome and, possibly, request a retransmission (NACK). This is received with no error by PUtx and overheard by the SU pair.

We assume all codewords are drawn from a Gaussian codebook, and are sufficiently long to allow reliable decoding whenever the attempted rate is within the mutual information rate of the channel. Let $C(\mathrm{SNR}){\triangleq}\log_{2}(1{+}\mathrm{SNR})$ be the capacity of the Gaussian channel as a function of the SNR at the receiver \cite{Cover}. We denote the packets being transmitted by the SU and PU with their labels $l_S$ and  $l_P$, respectively ($l_S{=}0$ if the SU remains idle). We now describe the SU system.
 \vdo{-3mm}
 \vone{-5mm}
\subsection{Decoding outcomes at SUrx}
The decoding performance at SUrx depends on whether $l_P$ is known or not at SUrx to perform interference cancellation, as a result of a previous successful decoding operation, and on the access decision $a_{S,t}{\in}\{0,1\}$, as detailed below.

 \subsubsection{Case $a_{S,t}{=}1$, $l_P$ unknown}
\label{subsecdecevent}
SUrx attempts to decode both $l_S$ and $l_P$. Since SUtx, PUtx and SUrx form a multiple access channel \cite{Cover}, the outcomes at the receiver for a given rate pair $(R_s,R_p)$, as a function of the SNRs $(\gamma_s,\gamma_{ps})$, are as depicted in  Fig.~\ref{decevents}. We denote their probabilities as
\begin{align}
\label{outcomes}
\hspace{-2mm}
\begin{array}{ll}
\text{\{1\}:}\ \delta_{sp}\triangleq\mathbb P(\text{$l_P$ \& $l_S$ decoded}), 
&\hspace{-3.5mm}\text{\{4\}:}\ \upsilon_{s}\triangleq\mathbb P(\text{$l_P\rightarrow l_S$}),
\\
\text{\{2\}:}\ \delta_{s}\triangleq\mathbb P(\text{only $l_S$ decoded}),
&\hspace{-3.5mm}\text{\{5\}:}\ \upsilon_{p}\triangleq\mathbb P(\text{$l_S\rightarrow l_P$}),
\\
\text{\{3\}:}\ \delta_{p}\triangleq\mathbb P(\text{only $l_P$ decoded}),
&\hspace{-3.5mm}\text{\{6\}:}\ \upsilon_{sp}\triangleq\mathbb P(\text{$l_S\leftrightarrow l_P$}),
\\
\text{\{7\}:}\ \upsilon_{\emptyset}\triangleq\mathbb P(\text{failure}),
\end{array}
\hspace{-3mm}
\end{align}
computed as the marginals with respect to the distribution of $(\gamma_s,\gamma_{ps})$.
In \{1\}, $l_S$ and $l_P$ are jointly decoded.
In \{2\} (respectively, \{3\}), only $l_S$ ($l_P$) is decoded, by treating the interfering $l_P$ ($l_S$) as noise. In \{4\}, \{5\} or \{6\}, neither $l_S$ nor  $l_P$ can be currently decoded by SUrx; however, one packet can be decoded only after removing the interference from the other.
 The arrow $l_X{\rightarrow}l_Y$ indicates the decoding dependence between $l_X$ and $l_Y$, so that $l_Y$ can be decoded after removing the interference from $l_X$, but not vice-versa (unless $l_X\leftrightarrow l_Y$). In these three cases, the received signal is buffered at SUrx for future recovery via chain decoding, see Sec.~\ref{buffering}. Finally, in \{7\}, the channel quality is poor, so that neither  $l_S$ nor  $l_P$ can be decoded by SUrx (even after removing their mutual interference) and the signal is discarded.
 
 \begin{figure}[t]
\centering  
\includegraphics[width=\w{.6}{.5}\linewidth,trim = 0mm 0mm 0mm 0mm,clip=false]{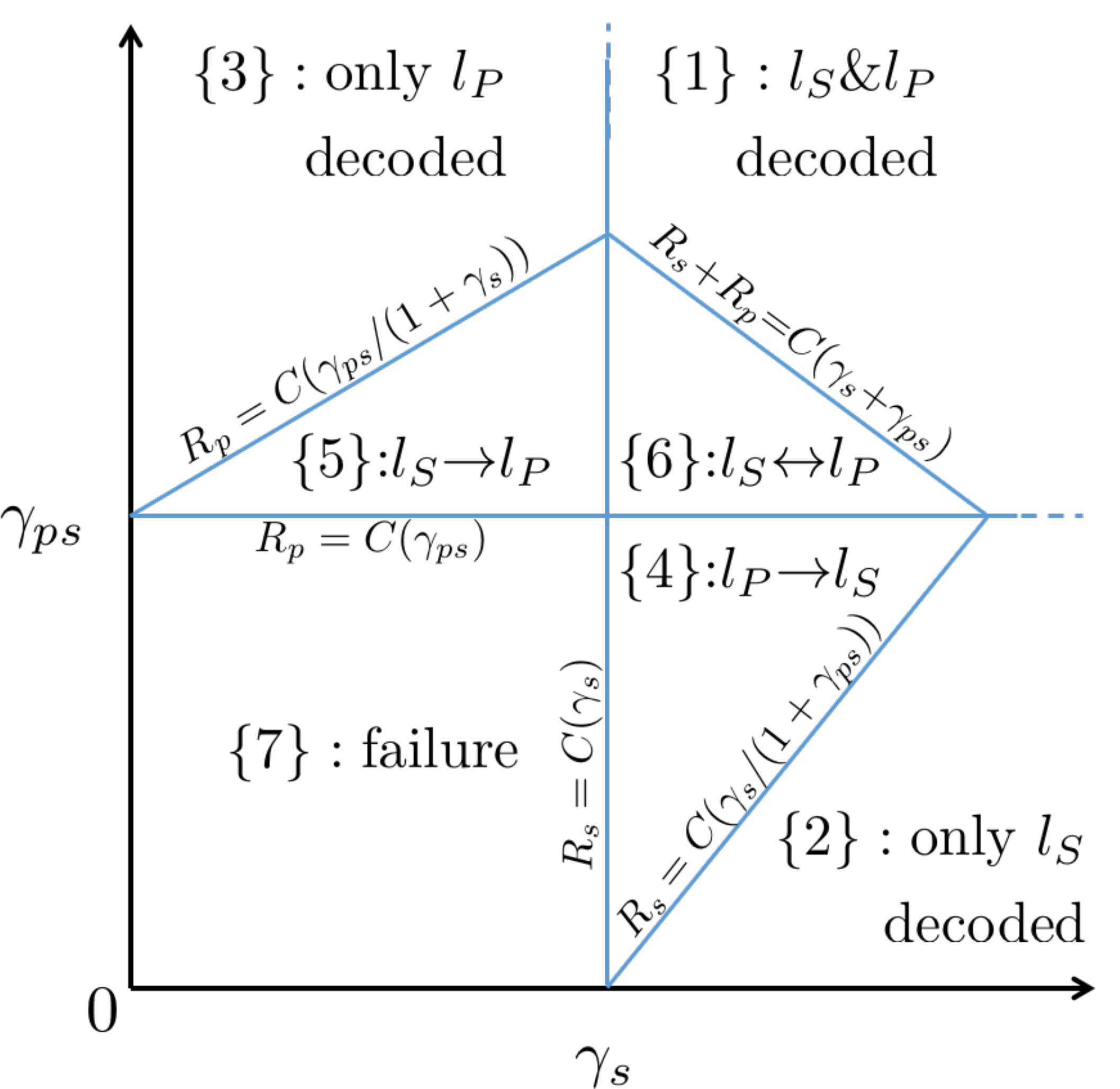}
\caption{Decoding regions at SUrx as a function of $(\gamma_s,\gamma_{ps})$. The boundaries correspond to the decoding thresholds of the multiple access channel \cite{Cover}.}
\label{decevents}
\vdo{-5mm}
\vone{-10mm}
\end{figure}

\subsubsection{Case $a_{S,t}{=}1$, $l_P$ known}
$l_P$ is known at SUrx as a result of a previous decoding operation at SUrx, its interference is removed from the received signal, thus creating an interference-free channel to decode $l_S$. Therefore, the SU transmission succeeds if $R_s< C(\gamma_{s,t})$. Since this event is the union of the four disjoint events $\{1\}$, $\{2\}$, $\{4\}$ and $\{6\}$ (right side of the decoding threshold $R_s=C(\gamma_{s})$ in Fig.~\ref{decevents}), its probability is obtained via (\ref{outcomes}) as
\begin{align}
\label{Ds}
D_s\triangleq\mathbb P(R_s< C(\gamma_{s,t}))=\delta_{s}+\delta_{sp}+\upsilon_{s}+\upsilon_{sp}.
\end{align}

\subsubsection{Case $a_{S,t}{=}0$, $l_P$ unknown}
SUtx remains idle and SUrx attempts to decode $l_P$. Thus, SUrx decodes $l_P$ successfully if $R_p< C(\gamma_{ps,t})$. Since this event is the union of the four disjoint events $\{1\}$, $\{3\}$, $\{5\}$ and $\{6\}$ (region above the decoding threshold $R_p=C(\gamma_{ps})$ in Fig.~\ref{decevents}), its probability is obtained via (\ref{outcomes}) as
\begin{align}
\label{Dp}
D_p\triangleq\mathbb P(R_p< C(\gamma_{ps,t}))=\delta_{p}+\delta_{sp}+\upsilon_{p}+\upsilon_{sp}.
\end{align}

\subsubsection{Case $a_{S,t}{=}0$, $l_P$ known}
no decoding activity at SUrx.
\vdo{-3mm}
\vone{-5mm}
\subsection{Decoding feedback from SUrx}
\label{SUFB}
At the end of each slot, SUrx feeds back $y_{S,t}\in\{1,\dots,7\}$ to SUtx over a dedicated error-free control channel, indicating one of the regions of Fig.~\ref{decevents} (the numbering is given in (\ref{outcomes}) as $\text{\{j\}}$). This feedback signal, together with the ARQ feedback signal received from PUrx, allows SUtx to keep track of the chain decoding state, the buffering of corrupted signals and the knowledge of the current PU packet $l_P$ at SUrx.
\vdo{-3mm} 
\vone{-5mm}
\subsection{SU retransmissions, buffering and chain decoding}
\label{buffering}
The SU performs retransmissions and buffering at SUrx to improve the potential of interference cancellation at SUrx. For instance, if  $l_S{\rightarrow}l_P$ or  $l_P{\leftrightarrow l_S}$ in a previous slot, the SU may retransmit $l_S$. If $l_S$ is decoded by SUrx, its interference can be removed from the previously buffered received signal to recover $l_P$. In turn, the recovered  $l_P$ may be exploited to recover other SU packets from previously buffered signals received within the ARQ window associated to $l_P$, via interference cancellation, see example in Fig.~\ref{figexlabel2}. The iterative application of interference cancellation on signals buffered at SUrx is denoted as \emph{chain decoding} (CD). Thus, when  $l_P{\rightarrow}l_S$,  $l_S{\rightarrow}l_P$ or  $l_P{\leftrightarrow}l_S$, with probability $\upsilon_{s}$, $\upsilon_{p}$ and $\upsilon_{sp}$, respectively, SUrx buffers the corresponding received signals. For analytical tractability, we assume an infinite buffer at SUrx. We will evaluate the effect of a finite buffer size numerically in Sec.~\ref{sec:numres}.
 
\begin{definition}
\label{defCD}
 The decoding relationship among the SU and PU packets buffered at SUrx is represented by the \emph{CD graph}, with vertices the set of undecoded packets, and edges the decoding relationship among them. For instance, if $l_{S}\rightarrow l_{P}$,
then $l_{S}$ and $l_{P}$ are vertices in the CD graph, connected by a directed edge from $l_{S}$ to  $l_{P}$. We define the \emph{CD root} as the SU packet which, once decoded, triggers the recovery of the largest number of SU packets  via CD (see Fig.~\ref{CDgraph}).
\end{definition}

 \begin{figure*}
    \centering
\includegraphics[scale=0.25]{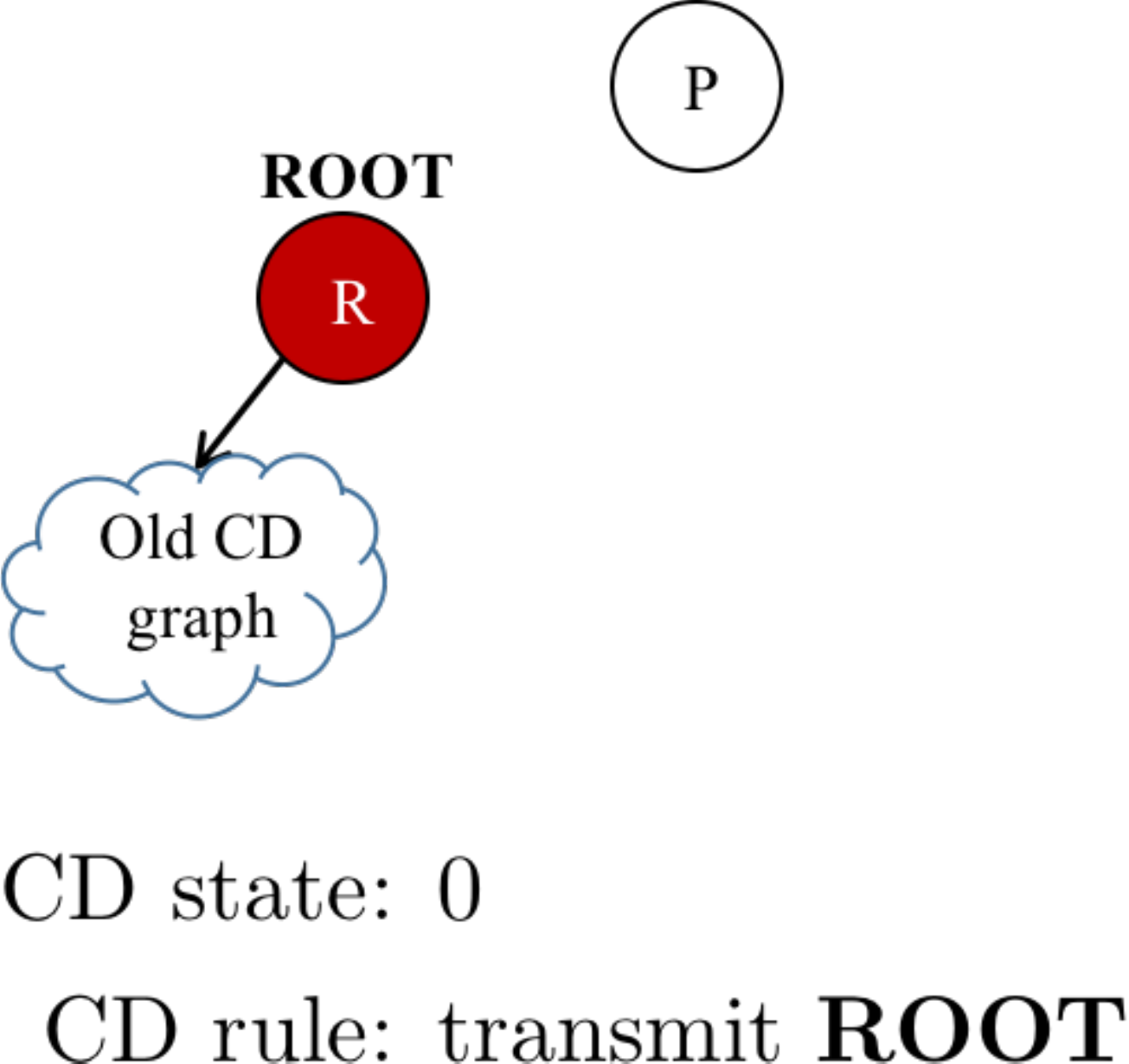}
\hspace{5mm}
\includegraphics[scale=0.25]{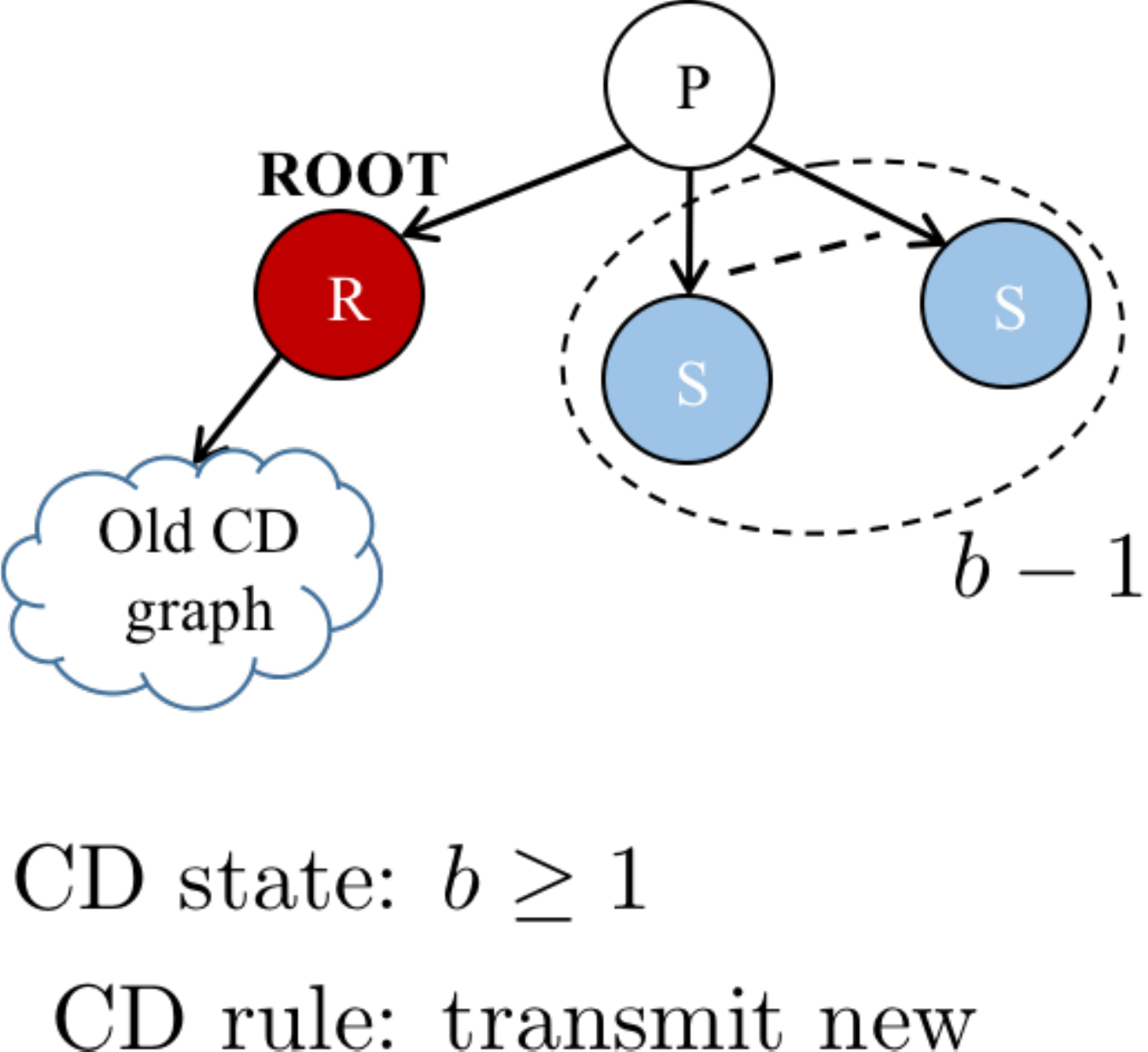}
\hspace{5mm}
\includegraphics[scale=0.25]{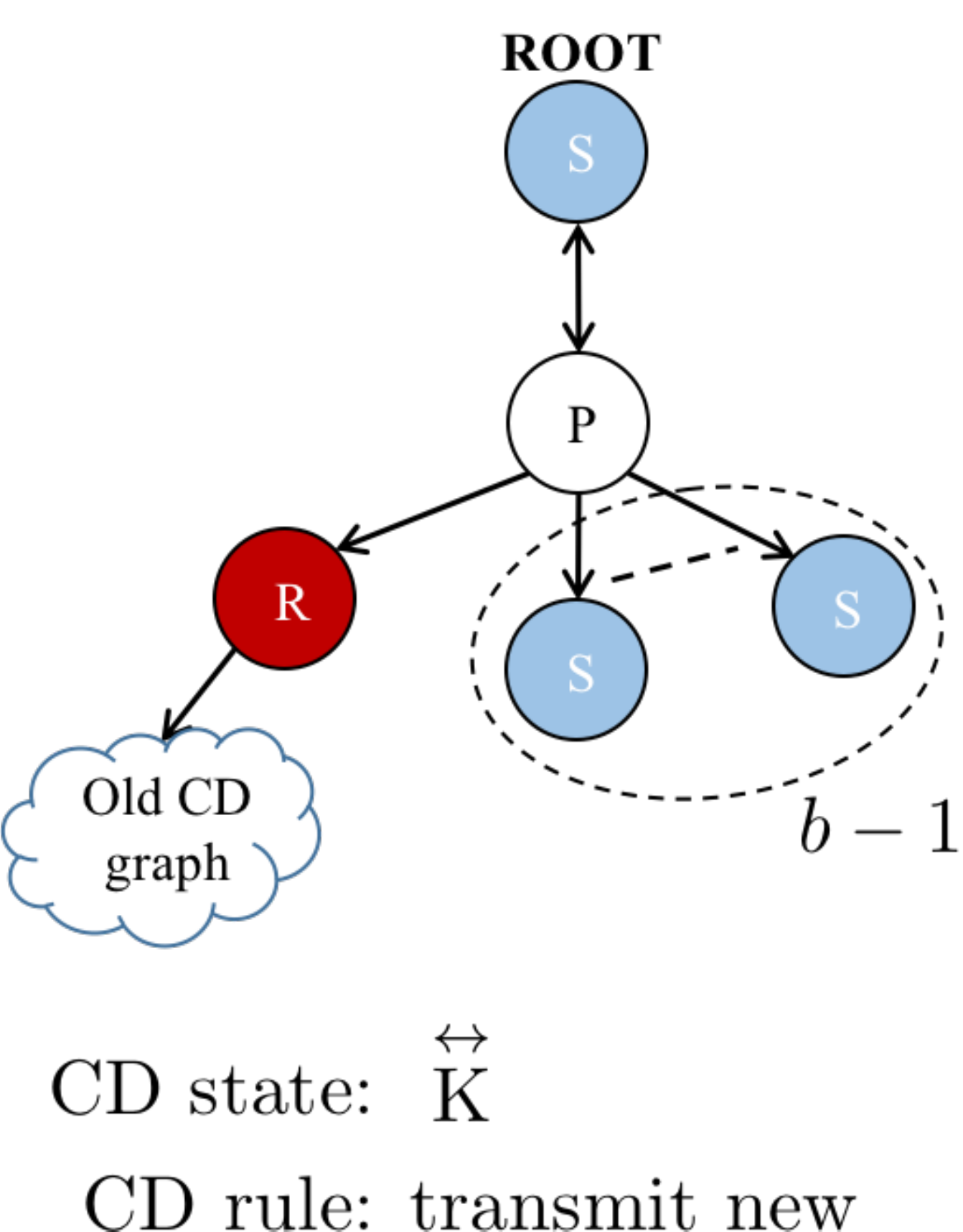}
\hspace{5mm}
\includegraphics[scale=0.25]{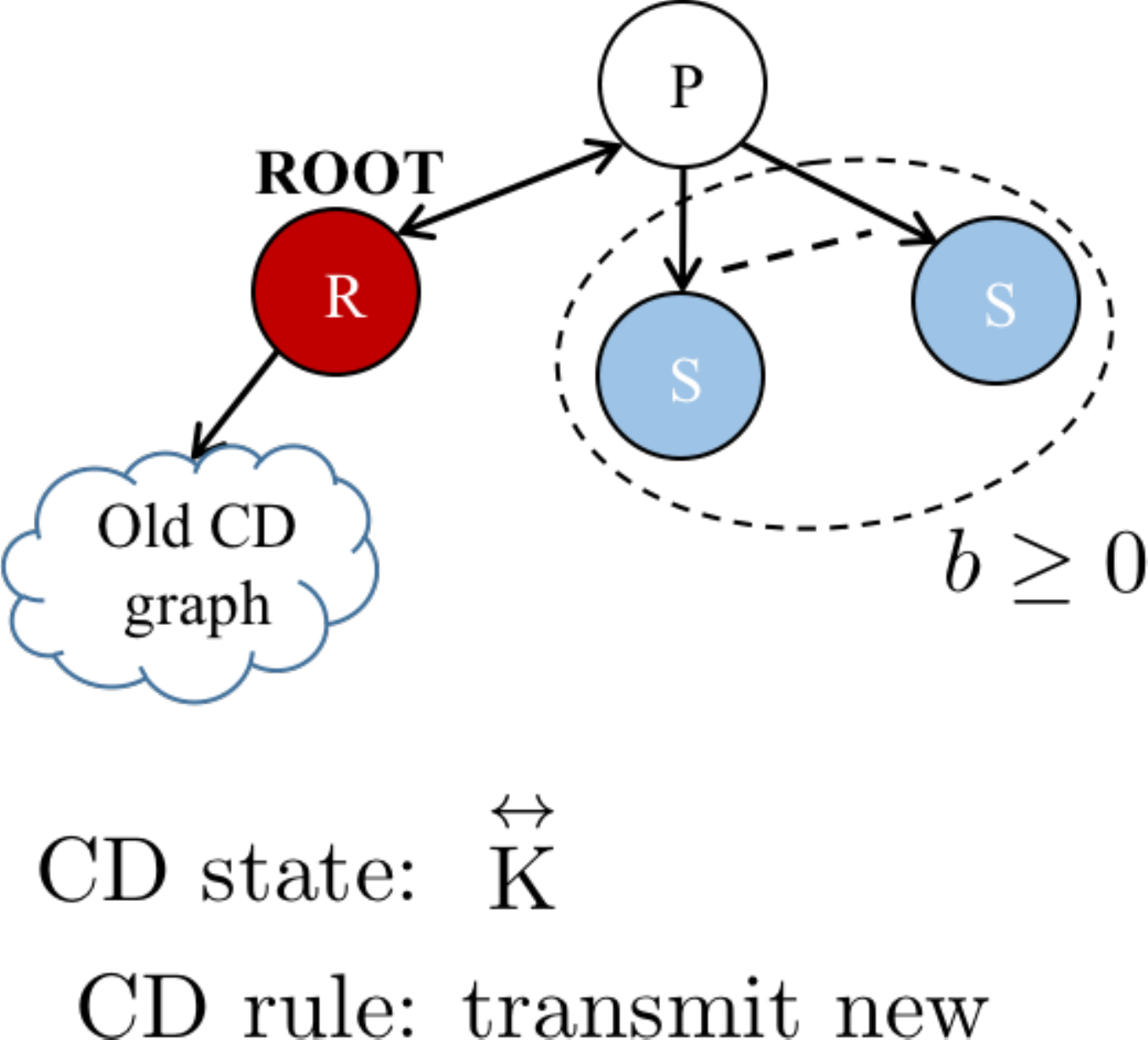}
\\
(Fig.~\ref{CDgraph}.a)
\hspace{25mm}
(Fig.~\ref{CDgraph}.b)
\hspace{25mm}
(Fig.~\ref{CDgraph}.c)
\hspace{25mm}
(Fig.~\ref{CDgraph}.d)
\\
\vspace{5mm}
\includegraphics[scale=0.25]{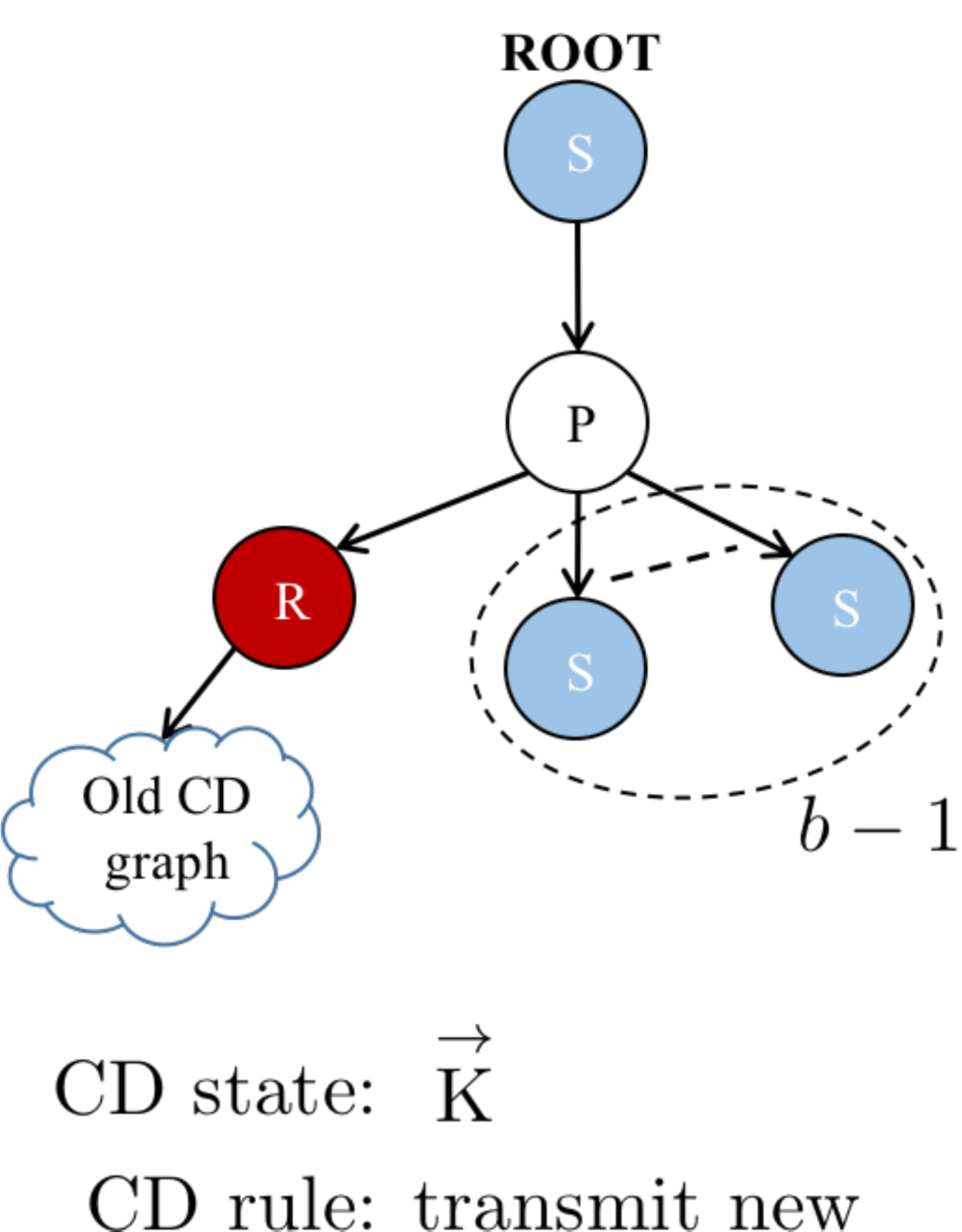}
\hspace{5mm}
\includegraphics[scale=0.25]{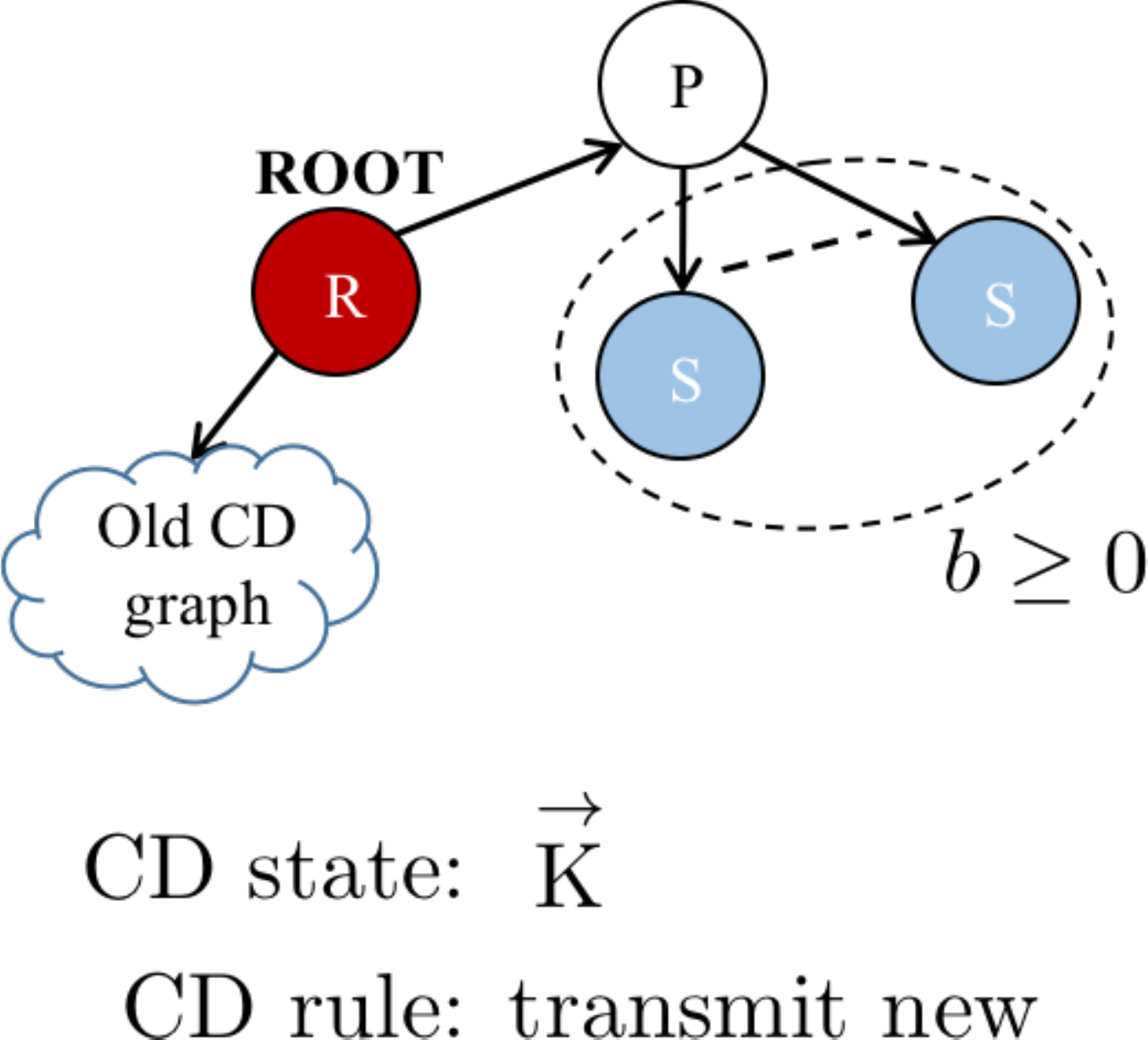}
\hspace{5mm}
\includegraphics[scale=0.25]{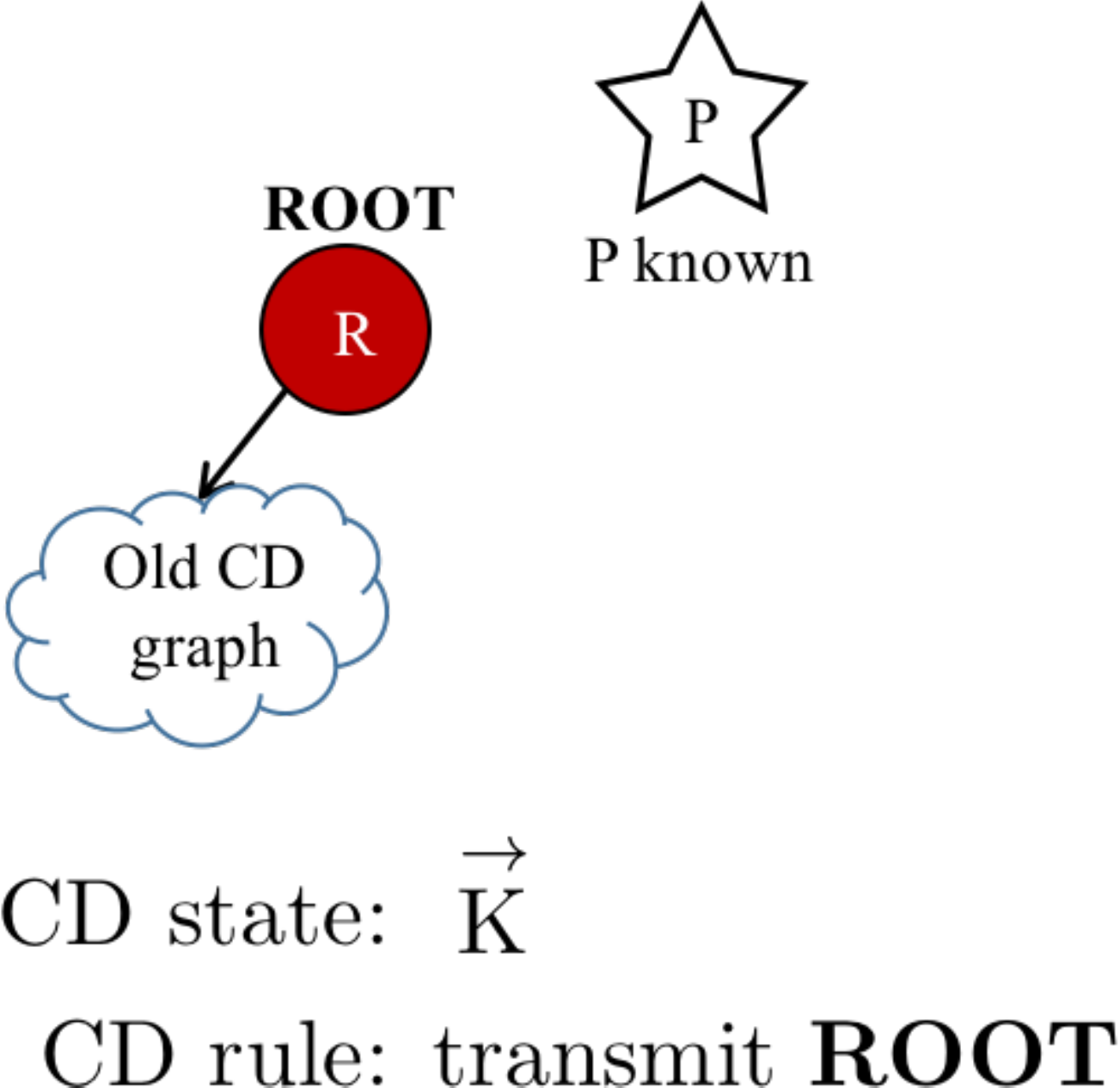}
\\
\vspace{2mm}
(Fig.~\ref{CDgraph}.e)
\hspace{25mm}
(Fig.~\ref{CDgraph}.f)
\hspace{25mm}
(Fig.~\ref{CDgraph}.g)
\caption{States of the CD graph. "P" denotes the packet currently transmitted by the PU;
"S" denotes SU packets; "\textbf{ROOT}" denotes the root of the CD graph (see Definition~\ref{defCD}); "Old CD graph" denotes the CD graph inherited from the previous ARQ window and still undecoded; similarly, "R" denotes its root.}
\label{CDgraph}
\vdo{-4mm}
\vone{-7mm}
\end{figure*}

The retransmission process at the SU is governed by the \emph{packet selection policy}:
  if $a_{S,t}{=}1$, it selects which SU packet to (re)transmit based on the structure of the CD graph at SUrx. In \cite{MichelusiCD}, we have shown that the optimal packet selection policy follows a chain decoding protocol, assumed in the rest of the paper. We refer to \cite{MichelusiCD} for details and proof of optimality. Herein, we describe the CD protocol with the help of Fig.~\ref{CDgraph}. At the start of a new ARQ window (the PU transmits a new packet), the PU packet is unknown at SUrx and the configuration is depicted in Fig.~\ref{CDgraph}.a. The CD graph evolves over the ARQ window, leading to one of the configurations in  Fig.~\ref{CDgraph}.a-g. In the configuration of Fig.~\ref{CDgraph}.a, the PU packet and the CD root are not connected: the CD protocol dictates to retransmit the CD root, so as to maximize the chances of either decoding it, or connecting it to the PU packet (when $l_P{\rightarrow}l_S$,  $l_S{\rightarrow}l_P$ or  $l_P{\leftrightarrow}l_S$), leading to one of the configurations in Fig.~\ref{CDgraph}.b,d,f. If the CD root is decoded, then the CD graph is decoded via chain decoding, along with the buffered packets.  If $l_P{\rightarrow}l_S$,  $l_S{\rightarrow}l_P$ or  $l_P{\leftrightarrow}l_S$, respectively, the new configuration becomes the one depicted in Fig.~\ref{CDgraph}.b with $b=1$, Fig.~\ref{CDgraph}.f with $b=0$ or Fig.~\ref{CDgraph}.d with $b=0$, respectively. Finally, if the PU packet is decoded, the new configuration becomes  the one depicted in Fig.~\ref{CDgraph}.g. Once the CD root and the PU packet are connected (Figs.~\ref{CDgraph}.b-f), it is optimal to transmit a new packet: this choice maximizes the chance of connecting it to the CD graph and leverage interference cancellation in the future; retransmitting the CD root would be redundant, since it is already connected to the CD graph. In the configuration of Fig.~\ref{CDgraph}.g,  the PU packet is known at SUrx, and thus its interference is cancelled. In this case, it is optimal to retransmit the CD root to maximize the chances of decoding the CD graph, by taking advantage of the interference-free channel.
  \vdo{-3mm}
  \vone{-5mm}
\subsection{SU access policy}
\label{sec:policies}
At the beginning of slot $t$, given the history up to slot $t$, $\mathcal H^t{=}(y_{P,0}^{t-1},y_{S,0}^{t-1},a_{S,0}^{t-1})$, SUtx selects $a_{S,t}{=}1$ with probability 
$\mu_{t}(\mathcal H^t)$, and $a_{S,t}{=}0$ otherwise, where $\mu_t$ denotes the access policy. If $a_{S,t}{=}1$, the CD protocol  described in Sec.~\ref{buffering} dictates which 
 packet to transmit (label $l_S$).
 \vdo{-2mm}
 \vone{-5mm}
\section{Optimization Problem}
\label{perfopt}
 We define the \emph{average long-term} PU throughput as
 \begin{align}
 \label{barTp}
 \bar T_P(\mu)\triangleq\lim_{D\to\infty}
R_p\mathbb E\left[ \frac{1}{D}\sum_{t=0}^{D-1}(1-\rho_{a_{S,t}})\right],
 \end{align}
 where $D$ is the horizon duration, the expectation is taken with respect to the sequence $\{a_{S,t},t{\geq}0\}$ generated by the access  policy $\mu$, and the decoding outcomes at SUrx and PUrx. This metric is equivalent to the "stable throughput," which guarantees stability in  a systems where packets generated in upper layers are stored in queues before transmission \cite{Kulkarni}. Since $\rho_{a_{S,t}}{=}\rho_{0}{+}a_{S,t}(\rho_{1}{-}\rho_{0})$, one can rewrite
  \begin{align}
  \label{equiv}
 \bar T_P(\mu)=
 \bar T_{P,\max}\Bigr[1-\nabla(\mu)\Bigr],
 \end{align}
 where $\bar T_{P,\max}{\triangleq}R_p(1{-}\rho_{0})$ is the maximum PU throughput, achievable when the SU remains always idle, and we have defined the PU throughput degradation, relative to the maximum throughput $\bar T_{P,\max}$, as 
  \begin{align}
  \label{nabla}
\nabla(\mu)\triangleq
\frac{\rho_{1}-\rho_{0}}{1-\rho_{0}}
\lim_{D\to\infty}
 \mathbb E\left[\frac{1}{D}\sum_{t=0}^{D-1}a_{S,t}\right].
 \end{align}
 We can interpret $\nabla(\mu)$ as the throughput loss experienced by the PU as a result of the activity of the SU, which should be limited to reflect higher layer QoS constraints \cite{Zhang}. Similarly, we define the average long-term SU throughput as
 \begin{align}
 \bar T_S(\mu)\triangleq\lim_{D\to\infty}
\mathbb E\left[ \frac{1}{D}\sum_{t=0}^{D-1}r_{S,t}\right],
 \end{align}
where $r_{S,t}$ is the instantaneous throughput accrued via CD.

The goal is to design the SU access policy $\mu$ so as to maximize $\bar T_S(\mu)$,
subject to a maximum PU throughput degradation constraint $\nabla_{th}{\in}(0,1)$ (alternatively, subject to a minimum PU throughput $\bar T_{P,\min}=\bar T_{P,\max}(1-\nabla_{th})$ via (\ref{equiv})),
\begin{align}
\label{optP1}
&\mathbf{OP:}\ \mu^*=\arg\max_{\mu} \bar T_S(\mu),\ \text{s.t.\ }\nabla(\mu)\leq\nabla_{th}.
\end{align}
Note from (\ref{nabla}) that $\nabla(\mu)$ is maximum when $a_{S,t}=1,\forall t$, yielding
$\nabla(\mu)\leq\nabla_{\max}\triangleq \frac{\rho_{1}-\rho_{0}}{1-\rho_{0}}$; then, if $\nabla_{th}\geq \nabla_{\max}$, the constraint in $\mathbf{OP}$ becomes inactive.

The SU throughput and interference are cumbersome to compute in this form, since the outcome of CD depends on the specific instance of the CD graph. As shown in \cite{MichelusiCD}, a simplification can be obtained using the concept of \emph{virtual decodability}.
\vdo{-2mm}
\begin{definition}
\label{virtualdecod}
A packet $l$ in the CD graph is \emph{virtually decodable} if it becomes decodable by initiating CD at the CD root, following the directed edges in the CD graph (CD root excluded).
Otherwise, we say it is \emph{virtually undecodable}.
\end{definition}
Based on this definition, if $l$ is virtually decodable and the CD root is decoded, then $l$ is also decoded via CD. Therefore, if one guarantees to decode the CD root with probability one, eventually any virtually decodable $l$ will also be decoded. Indeed, this is the case: according to the optimal CD rules \cite{MichelusiCD}, as explained in Sec.~\ref{buffering}, the CD root is retransmitted (at least) at the beginning of each ARQ window (Fig.~\ref{CDgraph}.a). Eventually, it will be decoded, triggering chain decoding over the entire CD graph; thus, $l$ can be considered \emph{virtually} decoded, even if it has not been \emph{currently} decoded. As a result, there is no loss of generality, in terms of average throughput, if one counts the  \emph{virtually decodable} packets at the present time, rather than at the future time when they are actually decoded via CD. Based on this intuition, in \cite{MichelusiCD} we have shown that $\bar T_S(\mu)$ can be expressed as
\begin{align}
\label{TSvirtual}
\bar T_{S}(\mu)=\lim\inf_{D\to\infty}\mathbb E\left[\frac{1}{D}\sum_{t=0}^{D-1}v_S(a_{S,t},\mathbf s_{t})\right],
\end{align}
where $v_S(a_{S,t},\mathbf s_{t})$ is the \emph{expected virtual instantaneous throughput} (which counts the virtually decoded SU packets in addition to the currently decoded ones),
whose analytical expression is provided in Sec.~\ref{transprob}, and $\mathbf s_t=(\Phi,b)$ is the state of the CD protocol:
\begin{itemize}[leftmargin=0.3cm]
 \item $\Phi$ denotes the \emph{virtual} knowledge of the current PU packet $l_P$ at SUrx,
 and takes values in the set $\{\stackrel{\leftrightarrow}{\K},\stackrel{\rightarrow}{\K},\U\}$.
"K" denotes that the current PU packet $l_P$ is \emph{virtually decodable} at SUrx
(\emph{i.e.}, either it has been decoded in a previous slot by SUrx as in Fig.~\ref{CDgraph}.g, or it is virtually decodable as in Fig.~\ref{CDgraph}.c-f); in contrast,  "U" denotes the complementary event that \emph{$l_P$ is virtually undecodable} (Fig.~\ref{CDgraph}.a-b).
 The unidirectional or bidirectional arrow above "K" indicates the type of edge connecting $l_P$ to the CD root. In particular,  $\Phi=\stackrel{\leftrightarrow}{\K}$ (Fig.~\ref{CDgraph}.c-d) indicates that $l_P$ and the CD root are mutually decodable after removing their respective interference, \emph{i.e.}, $l_P\leftrightarrow [\text{CD root}]$; $\Phi=\stackrel{\rightarrow}{\K}$ indicates that either $l_P$ is known (Fig.~\ref{CDgraph}.g), or it can be decoded via CD after decoding the CD root, but not vice versa ($[\text{CD root}]\rightarrow l_P$ but not $l_P\rightarrow [\text{CD root}]$, Fig.~\ref{CDgraph}.e-f).
 \item if $\Phi{=}\U$ ($l_P$ is virtually undecodable), $b$ denotes the number of \emph{virtually undecoded} SU packets $l_S^{[1]},l_S^{[2]},\dots,l_S^{[b]}$ transmitted within the current ARQ window of $l_P$, such that $l_P{\rightarrow}l_S^{[i]}$, as in Fig.~\ref{CDgraph}.b. If $l_P$ becomes virtually decodable in state $\U$, then its interference can  be virtually removed, and thus $l_S^{[i]},i=1,2,\dots,b$ become virtually decodable as well; in contrast, if $l_P$ is not virtually decoded within the end of its ARQ window, then $l_S^{[1]},\dots,l_S^{[b]}$ remain undecoded and are discarded, since $l_P$ will not be transmitted again. Reliability of these SU packets may be enforced via retransmissions requested by higher layer protocols. We set $b{=}0$ when $\Phi{\in}\{\stackrel{\leftrightarrow}{\K},\stackrel{\rightarrow}{\K}\}$, since $l_P$ is virtually decodable in these cases and the SU channel is, virtually, interference free.
\end{itemize} 
When a new ARQ window begins, the new PU packet is virtually undecodable and $b{=}0$ (Fig.~\ref{CDgraph}.a), hence the new state becomes $(\U,0)$.
\vdo{-6mm}
\vone{-5mm}
\section{Analysis}
\label{analysis}
Under this equivalent formulation, the operation of the SU is a Markov decision process (MDP) \cite{Bertsekas}, with state $\mathbf s_{t}\in\mathcal S$,\footnote{For compactness, we write state $(\U,b)$ as $b$, $(\stackrel{\leftrightarrow}{\K},0)$ as $\stackrel{\leftrightarrow}{\K}$, $(\stackrel{\rightarrow}{\K},0)$ as $\stackrel{\rightarrow}{\K}$.}
infinite (but countable) state space
\begin{align}
& \mathcal S=
\left\{b|b\geq 0\vphantom{\stackrel{\leftrightarrow}{\K},\stackrel{\rightarrow}{\K}}\right\}
\cup
\left\{\stackrel{\leftrightarrow}{\K},\stackrel{\rightarrow}{\K}
\right\},
\end{align}
action $a_{S,t}\in\{0,1\}$, reward $v_S(a_{S,t},\mathbf s_{t})$ (to compute the SU throughput (\ref{TSvirtual})) and cost $\frac{\rho_{1}-\rho_{0}}{1-\rho_{0}}a_{S,t}$ (to compute the PU throughput degradation (\ref{nabla})). Thus, the optimal solution of {\bf OP} is a \emph{stationary} and \emph{state-dependent} policy \cite{Ross1989}, $\mu_{t}(\mathcal H^t)=\mu(\mathbf s_{t}),\ \forall t$. We let $\mathcal U$ be the set of such policies, 
\begin{align}
\mathcal U\equiv\{\mu:\mathcal S\mapsto[0,1]\}.
\end{align}

The transition probabilities and rewards $v_S(a_{S,t},\mathbf s_{t})$ of the MDP are characterized in Sec.~\ref{transprob}. Then, in Sec.~\ref{analysis}-B, we investigate the 
optimal SU access policy.

\vdo{-3mm}
\subsection{Virtual throughput and transition probabilities}
\label{transprob}
In state $b$ (Fig.~\ref{CDgraph}.a-b), $v_S(a_{S},b)$ is given by\footnote{Note that the SU packets in the "Old CD graph" do not appear in the expression of the virtual throughput, since they have already been virtually decoded in previous ARQ windows.}
\begin{align}
\label{G1}
v_S(a_{S},b)
=&
R_s\bigr[a_{S}(\delta_{sp}+\delta_{s})+D_pb\bigr].
\end{align}
In fact, with probability $D_p$, $l_P$ becomes virtually decodable, along with the $b$ buffered SU packets connected to it; if $a_S{=}1$, $l_S$ is decoded successfully with probability $\delta_{sp}+\delta_{s}$, due to the interference from the PU signal.
\begin{remark}
Note that the probability of \emph{virtually} decoding $l_P$ is $D_p$, irrespective of whether SUtx transmits or remains idle; in fact, when SUtx transmits $l_S$, $l_P$ is decoded with probability  $\delta_{p}+\delta_{sp}$ (see Fig.~\ref{decevents}), and it is virtually decoded if $l_S\rightarrow l_P$ (with probability $\upsilon_p$) or $l_P\leftrightarrow l_S$ (with probability $\upsilon_{sp}$), yielding $D_p=\delta_{p}+\delta_{sp}+\upsilon_{p}+\upsilon_{sp}$ as the overall probability of (possibly, only virtually) decoding $l_P$. Then, the \emph{virtual} decodability of the PU packet is not hampered by the interference caused by the SU's own signal.
\end{remark}
In state $\stackrel{\leftrightarrow}{\K}$ (Fig.~\ref{CDgraph}.c-d),\footnote{Note that in the configurations of Fig.~\ref{CDgraph}.c-f, the $b$ SU packets such that $l_P\rightarrow l_S^{[1]},l_S^{[2]},\dots,l_S^{[b]}$ are not counted in the virtual throughput, since they have already been virtually decoded in the transitions leading to these configurations (e.g., from Fig.~\ref{CDgraph}.b to Fig.~\ref{CDgraph}.c).}
    \begin{align}
    \label{G2}
v_S(a_{S},\stackrel{\leftrightarrow}{\K})
&{=}
R_s\bigr[a_{S}(D_s-\upsilon_{sp})+D_p\bigr].
\end{align}
In fact, since $l_P$ is virtually decodable, $l_S$ can be decoded successfully with probability $D_s$ (since the channel is, virtually, interference-free), thus accruing the term $a_{S}D_s$.
With probability $[a_{S}(\delta_p{+}\delta_{sp})+(1{-}a_{S})D_p]$, $l_P$ is decoded;
it follows that the CD root is decoded (since $[\text{CD root}]{\leftrightarrow}l_P$ in state $\Phi{=}\stackrel{\leftrightarrow}{\K}$), thus accruing one unit of throughput. Finally, with probability $a_S\upsilon_p$, the transmission outcome is such that $l_S{\rightarrow}l_P$; it follows that $l_S$ becomes the new CD root leading to the new configuration of Fig.~\ref{CDgraph}.e, and the previous CD root is virtually decoded (since $[\text{previous CD root}]{\leftrightarrow}l_P$ in state $\Phi{=}\stackrel{\leftrightarrow}{\K}$), thus accruing one unit of throughput (see \cite{MichelusiCD}). We obtain (\ref{G2}) by adding up all these terms. Finally, in state $\stackrel{\rightarrow}{\K}$, 
      \begin{align}
      \label{G3}
v_S(a_{S},\stackrel{\rightarrow}{\K})
=&
a_{S}R_sD_s,
\end{align}
since $l_P$ is virtually decodable and the channel is (virtually) interference-free. We now derive the transition probabilities $P(\mathbf x|\mathbf s,a_S)\triangleq\mathbb P(\mathbf s_{t+1}{=}\mathbf x|\mathbf s_{t}{=}\mathbf s,a_{S,t}{=}a_S)$, by adapting those in \cite{MichelusiCD} to the model of this paper, with backlogged PU and infinite ARQ deadline.  
From state $b\geq 0$, 
\begin{align}
\label{tp1}
&
\!\!\!\!\!P(\mathbf x|b,a_S)
{=}
\left\{\begin{array}{ll}
1{-}\rho_{a_{S}}(D_p{+}a_{S}\upsilon_s), & \mathbf x{=}0,\text{ if }b=0,
\\
1-\rho_{a_{S}}, &\mathbf x{=}0,\text{ if }b>0,
\\
\rho_{a_{S}}\left(1{-}D_p{-}a_{S}\upsilon_s\right), &\mathbf x{=}b,\text{ if }b>0,
\\
\rho_1a_{S}\upsilon_s, & \mathbf x{=}b+1,
\\
\rho_1a_{S}\upsilon_{sp}, &\mathbf x{=}\stackrel{\leftrightarrow}{\K},
\\
\rho_{a_{S}}\left(D_p-a_{S}\upsilon_{sp}\right), &\mathbf x{=}\stackrel{\rightarrow}{\K}.
\end{array}\right.
\end{align}
In fact, the PU transmission succeeds with probability $1{-}\rho_{a_{S}}$; in this case,
a new ARQ window begins with a new PU packet, which is virtually undecodable to the SU, so that the new state becomes $\mathbf x{=}0$. If the transmission outcome is such that $l_P{\rightarrow}l_S$ and the PU fails, then the signal is buffered and $b$ increases by one unit, thus the state becomes $\mathbf x{=}b{+}1$. If the PU fails and $l_P$ is virtually decoded (with probability  $\rho_{a_{S}}D_p$), then the new state becomes $\mathbf x{=}\stackrel{\rightarrow}{\K}$ or $\mathbf x{=}\stackrel{\leftrightarrow}{\K}$, depending on whether $l_S{\rightarrow}l_P$ (with probability $D_p{-}a_{S}\upsilon_{sp}$) or $l_S{\leftrightarrow}l_P$ (with probability $a_{S}\upsilon_{sp}$), respectively. Otherwise, the state remains $\mathbf x{=}b$. From state $\stackrel{\leftrightarrow}{\K}$, 
\begin{align}
\label{tp3}
&
\!\!\!P(\mathbf x|\stackrel{\leftrightarrow}{\K},a_S)
{=}
\left\{\begin{array}{ll}
\!1-\rho_{a_{S}},&\!\!\!\mathbf x=0,
\\
\!\rho_{a_{S}}\left(1{-}D_p{+}a_{S}\upsilon_{sp}\right),&\!\!\!\mathbf x=\stackrel{\leftrightarrow}{\K},
\\
\!\rho_{a_{S}}\left(D_p-a_{S}\upsilon_{sp}\right), &\!\!\!\mathbf x=\stackrel{\rightarrow}{\K}.
\end{array}\right.
\!\!
\end{align}
In fact, the PU transmission succeeds with probability $1{-}\rho_{a_{S}}$ and the new state becomes $\mathbf x=0$. If the PU fails and the decoding outcome is such that $l_S\rightarrow l_P$, $l_S$ becomes the new CD root, and the new state becomes $\mathbf x=\stackrel{\rightarrow}{\K}$. Otherwise, the state does not change. From state $\stackrel{\rightarrow}{\K}$, 
\begin{align}
\label{tp4}
&P(\mathbf x|\stackrel{\rightarrow}{\K},a_S)
=
\left\{\begin{array}{ll}
1-\rho_{a_{S}},& \mathbf x=0,
\\
\rho_{a_{S}},& \mathbf x=\stackrel{\rightarrow}{\K}.
\end{array}\right.
\end{align}
In fact, with probability $1{-}\rho_{a_{S}}$ the PU succeeds and a new ARQ window begins. Otherwise, the state does not change.
\vdo{-5mm}
\vone{-5mm}
\subsection{Optimal SU access policy}
\label{struct}
In this section, we derive the  optimal SU access policy $\mu^*$ and its performance
in closed form. The main result is given in Theorems~\ref{thm:policy1} and~\ref{thm:policy2},
whose proof is provided in Sec.~\ref{sec:mainproof}. We let $\bar T_S^{(GA)}(\epsilon)\triangleq\epsilon R_sD_s$ be the \emph{genie-aided SU throughput} when SUrx has non-causal knowledge of the PU packets and can remove their interference (hence the success probability is $D_s$ in each slot), and SUtx transmits with probability $\epsilon=\min\left\{\frac{\nabla_{th}}{\nabla_{\max}},1\right\}$ to attain the constraint $\nabla_{th}$.
We let $\pi_{\mu}$ be the steady-state probability of the MDP under policy $\mu$.

Being genie-aided, $\bar T_S^{(GA)}(\epsilon)$ is an upper bound to the SU throughput.
A simple scheme to attain it is as follows:  SUtx remains idle until SUrx decodes the PU packet (state $\stackrel{\rightarrow}{\K}$); hence, it transmits with probability $\mu(\stackrel{\rightarrow}{\K})$ until the end of the ARQ window. By transmitting only in state $\stackrel{\rightarrow}{\K}$ when $l_P$ is known at SUrx, the genie-aided throughput $\bar T_S^{(GA)}$ is attained since SUrx can remove the interference of $l_P$ from the received signal, as in the genie-aided  case. If the PU throughput degradation constraint $\nabla_{th}$ augments, the access probability in state $\stackrel{\rightarrow}{\K}$ may be increased accordingly, so as to accrue larger SU throughput, until it becomes $\mu(\stackrel{\rightarrow}{\K}){=}1$. At this point, transitions from state $0$ (where the SU remains idle) to state $\stackrel{\rightarrow}{\K}$ occur with probability $\rho_0D_p$ (SUrx decodes the PU packet and the PU requests a retransmission); transitions from state $\stackrel{\rightarrow}{\K}$  (where the SU transmits) to state $0$ occur with probability $(1{-}\rho_1)$ (the PU succeeds and a new ARQ window begins), hence $\pi_\mu(\stackrel{\rightarrow}{\K}){=}\rho_0D_p/[1{-}\rho_1{+}\rho_0D_p]$ at steady-state
and the SU transmits over a fraction $\pi_\mu(\stackrel{\rightarrow}{\K})$ of the slots, yielding
the PU throughput degradation
\begin{align}
\label{NAGA}
\nabla_{GA}\triangleq
\pi_\mu(\stackrel{\rightarrow}{\K})\nabla_{\max}
{=}\frac{\rho_0D_p}{1-\rho_1+\rho_0D_p}.
\end{align}
This result is summarized in the following theorem.
\begin{thm}
\label{thm:policy1}
If $\nabla_{th}{\leq}\nabla_{GA}$, then
\begin{align}
\label{pol1}
\left\{\begin{array}{l}
\mu^*(0)=0,\\
\mu^*(\stackrel{\leftrightarrow}{\K})=\mu^*(b)=1,\ \forall b>0,\\
\mu^*(\stackrel{\rightarrow}{\K})=
\frac{[1-\rho_0(1-D_p)]\nabla_{th}}{
[1-\rho_0(1-D_p)]\nabla_{GA}
+(\rho_{1}-\rho_{0})(\nabla_{th}-\nabla_{GA})};
\end{array}\right.
\end{align}
under such policy,
\begin{align}
\label{GAth}
\bar T_S(\mu^*)=\bar T_S^{(GA)}\left(\frac{\nabla_{th}}{\nabla_{\max}}\right),
\qquad
\nabla(\mu^*)=\nabla_{th}.
\end{align}
\end{thm}
\begin{proof}
(\ref{GAth}) shows that the genie-aided throughput $\bar T_S^{(GA)}(\epsilon)$ is achievable under policy (\ref{pol1}) when $\nabla_{th}{\leq}\nabla_{GA}$. Indeed, since $\mu^*(0){=}0$, from (\ref{tp1}) with $b=0$ it follows that the transition probability to states $\stackrel{\leftrightarrow}{\K}$ and $b>0$ is zero, yielding $\pi_{\mu^*}(\stackrel{\leftrightarrow}{\K}){=}0$, $\pi_{\mu^*}(1){=}0$ and, by induction, $\pi_{\mu^*}(b){=}0,\forall b{>}0$. Therefore, SUtx never accesses states $\stackrel{\leftrightarrow}{\K}$ and $b{>}0$; in other words, it remains silent until the PU packet $l_P$ is decoded at SUrx (state $\stackrel{\rightarrow}{\K}$), as in the genie-aided case.
\end{proof}
%\vdo{-5mm}
Policy (\ref{pol1}) is randomized only in state $\stackrel{\rightarrow}{\K}$. By the property of MDPs \cite{Bertsekas}, the same performance is achieved by a policy that selects probabilistically (or time-shares between) one of the following two modes of operation at the beginning of each ARQ window (in the recurrent state $0$): \underline{Idle}: \emph{The SU remains idle over the entire ARQ window}; \underline{IC} (interference cancellation): \emph{The SU transmits (with probability one) only after the current PU packet is decoded at SUrx}. With Idle mode, the SU does not interfere at all with the PU; with IC mode, it leverages knowledge of the PU packet to perform interference cancellation, in the event that the SU packet is decoded at SUrx. In the limit $\nabla_{th}{\to}0$, the SU selects Idle mode with probability $\xi_1{=}1$. When $\nabla_{th}{=}\nabla_{GA}$, the SU selects IC mode with probability $\xi_2{=}1$. When $0{<}\nabla_{th}{<}\nabla_{GA}$, the probabilities $\xi_1$ and $\xi_2{=}1{-}\xi_1$ are chosen so as to attain the PU throughput degradation constraint with equality. 

When $\nabla_{th}{>}\nabla_{GA}$, the SU access probability in state $\stackrel{\rightarrow}{\K}$ can no longer be increased; therefore, higher SU throughput can only be achieved by transmitting  in state $0$ as well. The optimal policy for this case is determined in the following theorem.
\begin{thm}
\label{thm:policy2}
If $\nabla_{GA}{<}\nabla_{th}<\nabla_{\max}$, then
\begin{align}
\label{pol2}
\!\!\!\left\{\begin{array}{l}
\mu^*(0)
=
\frac{\nabla_{th}-\nabla_{GA}}
{
\nabla_{\max}-\nabla_{GA}
+(\nabla_{\max}-\nabla_{th})\frac{(\rho_1-\rho_0)D_p+\rho_1\upsilon_s}{1-\rho_1+\rho_0D_p}
},
\\
\mu^*(\stackrel{\rightarrow}{\K})=\mu^*(\stackrel{\leftrightarrow}{\K})=\mu^*(b)=1,\ \forall b>0;
\end{array}\right.
\end{align}
under such policy,
\begin{align}
\label{perf2}
\nonumber
&\bar T_S(\mu^*)=
\bar T_S^{(GA)}\left(\frac{\nabla_{th}}{\nabla_{\max}}\right)
%\\&\qquad
{-}
\frac{\rho_0D_p(1{-}\rho_1)\zeta R_s}{1-\rho_1(1-D_p)}
\frac{\nabla_{th}{-}\nabla_{GA}}{\nabla_{GA}\nabla_{\max}},
\\&
\nabla(\mu^*)=\nabla_{th},
\end{align}
where we have defined
\begin{align}
\label{zeta}
\zeta\triangleq \frac{\upsilon_{sp}}{1-\rho_1(1-D_p+\upsilon_{sp})}+\frac{\upsilon_s}{1-\rho_1(1-D_p)}.
\end{align}
Finally, if  $\nabla_{th}{\geq}\nabla_{\max}$, the "always transmit" policy
$\mu^*(\mathbf s){=}1,\forall \mathbf s\in\mathcal S$ is optimal, and
\begin{align}
\label{perf3}
&\bar T_S(\mu^*)=
\bar T_S^{(GA)}(1)
-\frac{(1-\rho_1)^2}{1-\rho_1(1-D_p)}\zeta R_s,
\\&
\nabla(\mu^*){=}\nabla_{\max}.
\end{align}
\end{thm}
 
If $\nabla_{GA}{<}\nabla_{th}{<}\nabla_{\max}$, SUtx transmits with non-zero probability in state $0$ until it reaches one of the states $1$, $\stackrel{\leftrightarrow}{\K}$ or $\stackrel{\rightarrow}{\K}$, see (\ref{pol2}). From this point on, it transmits with probability one until the end of the ARQ window. This policy is randomized only in state $0$. By the property of MDPs \cite{Bertsekas}, the same performance is achieved by a policy that selects probabilistically (or time-shares between) one of the following two modes of operation at the beginning of each ARQ window: IC mode as before; \underline{Always-TX}: \emph{The SU always transmits over the entire ARQ window}. With Always-TX mode, the SU maximizes the number of SU packets transmitted over the ARQ window and builds up the CD graph; these packets may become decodable via CD, hence this strategy maximizes the aggregate throughput accrued via CD. In the limit $\nabla_{th}{\to}\nabla_{GA}$, $\mu^*(0){\to}0$ and the SU selects IC mode with probability $\xi_2{=}1$. In the limit $\nabla_{th}{\to}\nabla_{\max}$, $\mu^*(0){\to}1$  and the SU selects  Always-TX mode with probability $\xi_3{=}1$. When $\nabla_{GA}{<}\nabla_{th}{<}\nabla_{\max}$, $\xi_2$ and $\xi_3{=}1{-}\xi_2$ are chosen so as  to attain the PU throughput degradation constraint with equality.
Finally, if $\nabla_{th}\geq\nabla_{\max}$, the constraint becomes inactive and the SU  selects  Always-TX mode deterministically so as to maximize the benefits of CD.

Overall, the optimal policy $\mu^*$ reflects a randomization among Idle, IC and Always-TX modes, with probabilities $\xi_1$, $\xi_2$ and $\xi_3=1{-}\xi_1{-}\xi_2$, respectively. If $\nabla_{th}{\leq}\nabla_{GA}$, $\xi_3{=}0$ so that only Idle and IC modes are used; if $\nabla_{GA}{<}\nabla_{th}{<}\nabla_{\max}$, $\xi_1{=}0$ so that only IC and Always-TX modes are used; finally, if $\nabla_{th}{\geq}\nabla_{\max}$,  $\xi_3=1$ so that only Always-TX mode is used
and the SU throughput is maximized.
\vone{-5mm}
\vdo{-3mm}
\subsection{Proof of Theorem~\ref{thm:policy2}}
\label{sec:mainproof}
We use a geometric approach inspired by \cite{MichelusiJSAC}, based on the properties of constrained MDPs \cite{Altman,Ross1989} to determine, in closed form, the optimal policy and its performance when $\nabla_{th}> \nabla_{GA}$. We make the following definition.
\begin{definition}[Deterministic/randomized policy]
A policy $\mu\in\mathcal U$ is deterministic if $\mu(\mathbf s)\in\{0,1\},\forall\mathbf s\in\mathcal S$; otherwise, $\mu$ is randomized. We let $\mathcal D\subset\mathcal U$ be the set of deterministic policies.
 \end{definition}
 In other words, $\mu\in\mathcal D$ takes a deterministic action in each state; however, the state sequence is random and governed by the transition probabilities under $\mu$, see (\ref{tp1})-(\ref{tp4}).
  \begin{figure}[t]
\centering  
\includegraphics[width=\w{.8}{.5}\linewidth,trim = 0mm 0mm 0mm 0mm,clip=false]{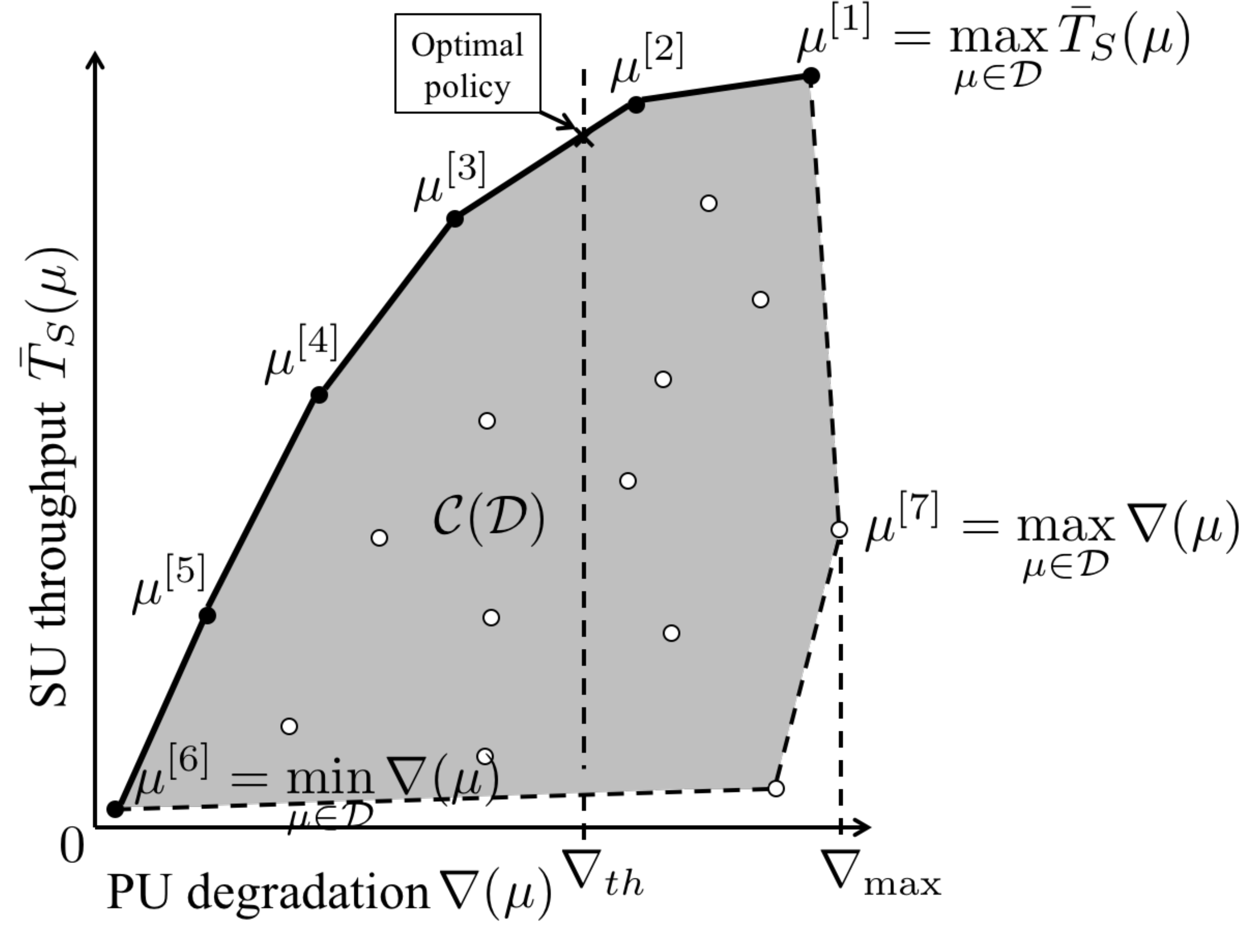}
\caption{Graphic representation of $\mathbf{OP}$.}
\vone{-10mm}
\vdo{-5mm}
\label{fig:txefficiency}
\end{figure}
Each deterministic policy $\mu\in\mathcal D$ attains a black or white circle in Fig.~\ref{fig:txefficiency}, located at coordinates $(\nabla(\mu),\bar T_S(\mu))$. The set of randomized policies, instead, attain the convex hull of the points with coordinates given by all the deterministic policies,\footnote{In fact, a randomized policy can be expressed equivalently as a time-sharing among deterministic policies.} denoted as $\mathcal C(\mathcal D)$, see Fig.~\ref{fig:txefficiency}. Thus, any point in $\mathcal C(\mathcal D)$
 can be achieved by a (possibly) randomized policy;  any point outside of $\mathcal C(\mathcal D)$ is, instead, unattainable.

According to {\bf OP}, the goal is, for a given PU throughput degradation constraint $\nabla_{th}$, to determine the optimal point $(\nabla(\mu^*),\bar T_S(\mu^*))$ within the convex hull $\mathcal C(\mathcal D)$, and the corresponding optimal policy $\mu^*$ maximizing the SU throughput $\bar T_S(\mu)$, see Fig.~\ref{fig:txefficiency}.
  We define the \emph{Pareto optimal envelope} of the convex hull $\mathcal C(\mathcal D)$, denoted as $\mathrm{PO}(\mathcal D)$,  as the set of points such as no improvement in the SU throughput $\bar T_S(\mu)$ is possible without causing additional degradation $\nabla(\mu)$ to the PU. This is indicated by the sequence of solid lines connecting the black circles in the figure. Mathematically,
 \begin{align}
 \nonumber
 &\mathrm{PO}(\mathcal D)=\bigr\{\vphantom{\sum}(\nabla(\mu),\bar T_S(\mu)):\mu\in\mathcal U, \forall \tilde\mu\in\mathcal U 
 \\&
 \text{ s.t. } \bar T_S(\tilde\mu)>\bar T_S(\mu)\Rightarrow \nabla(\tilde\mu)>\nabla(\mu)\bigr\}.
 \end{align}
  Accordingly, we define the set of Pareto optimal policies as
\begin{align}
\mathcal U_{\mathrm{PO}}\triangleq\bigr\{\mu\in\mathcal U:(\nabla(\mu),\bar T_S(\mu))\in\mathrm{PO}(\mathcal D)\bigr\}.
\end{align}
Note that any non Pareto optimal policy $\mu\notin\mathcal U_{\mathrm{PO}}$ is suboptimal since there exists $\tilde\mu\in\mathcal U$ such that $\bar T_S(\tilde\mu)\geq\bar T_S(\mu)$ and $\nabla(\tilde\mu)<\nabla(\mu)$. Thus, we have the following result.
 \begin{lemma}
 The optimal policy  is such that $\mu^*\in\mathcal U_{\mathrm{PO}}$.
 \end{lemma}
 
Hence, we can limit the search of the optimal policy within the set $\mathcal U_{\mathrm{PO}}$, which we aim to characterize.  Note that $ \mathrm{PO}(\mathcal D)$ is defined by a sequence of segments,  each with endpoints defined by a pair of deterministic policies $\mu^{[i]}$ and $\mu^{[i+1]}$, $i\geq 1$. Without loss of generality, the sequence $\{\mu^{[i]},i\geq 1\}$ is characterized by strictly decreasing values of the interference,
 $\nabla(\mu^{[i]})>\nabla(\mu^{[i+1]}),\forall i\geq 1$, and  $\mu^{[1]}$ is the deterministic policy which maximizes the  SU throughput (unconstrained), $\mu^{[1]}{=}\arg\max_{\mu\in\mathcal D}\bar T_S(\mu)$, derived in Lemma~\ref{Lemmamu1}.
Such sequence exhibits the following property.
 \begin{lemma}
$\{\mu^{[i]},\ i\geq 1\}$ defines
strictly decreasing values of $\nabla_{th}(\mu)$ and $\bar T_S(\mu)$, \emph{i.e.},
 \begin{align}
 \label{ordering}
\nabla_{th}(\mu^{[i]})>\nabla_{th}(\mu^{[i+1]}),
\quad
 \bar T_S(\mu^{[i]})>\bar T_S(\mu^{[i+1]}).
 \end{align}
 \end{lemma}
 \begin{proof}
If the above condition is not satisfied, \emph{i.e.}, $\nabla(\mu^{[i]})>\nabla(\mu^{[i+1]})$ and
$\bar T_S(\mu^{[i]})\leq\bar T_S(\mu^{[i+1]})$, then we achieve a contradiction on the Pareto optimality of $\mu^{[i]}$.
 \end{proof}

It follows that, given $\nabla_{th}>0$, one can determine $\mu^*$ and its performance as follows:
 \begin{itemize}[leftmargin=0.3cm]
 \item If $\nabla_{th}\geq\nabla(\mu^{[1]})$, then $\mu^*=\mu^{[1]}$; in fact, $\mu^{[1]}$ achieves the maximum unconstrained throughput and is feasible for the given value of $\nabla_{th}$;
\item Otherwise, let $i^*\geq 1$ be the unique index such that
$\nabla(\mu^{[i^*]})\geq\nabla_{th}>\nabla(\mu^{[i^*+1]})$; then, the optimal policy is given by a proper randomization (or time-sharing) between $\mu^{[i^*]}$ and $\mu^{[i^*+1]}$; 
we will characterize the form of this randomization throughout the proof.
 \end{itemize}
 
To characterize $ \mathrm{PO}(\mathcal D)$, we are left with the problem of finding the sequence $\{\mu^{[i]},i{\geq}1\}\subseteq\mathcal D$. To this end, we let $\mathcal D^{[i]}{\subseteq}\mathcal D$ be the set  of deterministic policies that interfere strictly less  than $\nabla(\mu^{[i]})$. Mathematically,
\begin{align}
\label{Ui}
\mathcal D^{[i]}\equiv
\left\{
\mu\in\mathcal D:\nabla(\mu)<\nabla(\mu^{[i]})
\right\}.
\end{align}
Then, by construction, $\mu^{[i+1]}$ is the deterministic policy which minimizes the slope of the segment connecting $(\nabla(\mu^{[i]}),\bar T_S(\mu^{[i]}))$ to $(\nabla(\mu),\bar T_S(\mu))$ over $\mu\in\mathcal D^{[i]}$, \emph{i.e.},
\begin{align}
\label{problemslope}
\mu^{[i+1]}=\arg\min_{\mu\in\mathcal D^{[i]}}\frac{\bar T_S(\mu)-\bar T_S(\mu^{[i]})}{\nabla(\mu)-\nabla(\mu^{[i]})},\ \forall i\geq 1.
\end{align}
In other words, $\mu^{[i+1]}$ is the deterministic policy that yields the minimum \emph{decrease} in SU throughput, relative to the decrease in PU throughput degradation.

Using this algorithm, we now determine $\mu^{[1]}$ and $\mu^{[2]}$. Lemma~\ref{Lemmamu1} states that $\mu^{[1]}$ is the Always-TX mode discussed in Sec.~\ref{struct}, and that it uniquely maximizes the interference $\nabla(\mu)$. That the Always-TX policy maximizes $\bar T_S(\mu)$ is an intuitive, but non trivial result; indeed on a setting without CD, it was  proved that Always-TX is not the throughput maximizing policy, see \cite{MichelusiITA11}. Then, Lemma~\ref{mu2} states that $\mu^{[2]}$ is the IC policy discussed in Sec.~\ref{struct}. It follows that, when $\nabla(\mu^{[1]})\geq\nabla_{th}>\nabla(\mu^{[2]})$, the optimal policy is obtained by time-sharing between the Always-TX policy $\mu^{[1]}$ and the IC policy $\mu^{[2]}$; alternatively, since the Always-TX and IC policies  differ only in state $0$, the same result is obtained by randomizing in state $0$,
yielding (\ref{pol2}).
\begin{lemma}
\label{Lemmamu1}
$\mu^{[1]}$ is \emph{uniquely} given by the Always-TX policy
\begin{align}
\mu^{[1]}(\mathbf s)=1,\ \forall \mathbf s\in\mathcal S.
\end{align}
Moreover, $\mathcal D^{[1]}\equiv\mathcal D\setminus\{\mu^{[1]}\}$.
\end{lemma}
\begin{proof}
\iftoggle{arxiv}{See \arxiv{Appendix B}.}{Due to space constraints, the proof is provided in \cite[Appendix B]{Proofs}.}
\end{proof}
\vdo{-3mm}
Given $\mu^{[1]}$ we now determine $\mu^{[2]}$ as the solution of the optimization problem (\ref{problemslope}). However, there is no need to minimize over the entire set  $\mathcal D^{[1]}\equiv \mathcal D\setminus\{\mu^{[1]}\}$. In fact, since {\bf OP} has one constraint,
the optimal policy is randomized in at most one state \cite{Altman}. Hence, any point in the segment connecting $(\nabla(\mu^{[1]}),\bar T_S(\mu^{[1]}))$ to $(\nabla(\mu^{[2]}),\bar T_S(\mu^{[2]}))$ is achievable by a policy randomized in at most one state, so that $\mu^{[1]}$ and $\mu^{[2]}$ differ in only one state. Letting $\mathbf s^{[1]}$ be such state, and
\begin{align}
\Delta_{\hat{\mathbf s}}(\mathbf s)=\chi(\mathbf s=\hat{\mathbf s}),\ \forall \mathbf s\in\mathcal S,
\end{align}
where $\chi(\cdot)$ is the indicator function, we can express $\mu^{[2]}$ as
\begin{align}
\label{sfghaf}
\mu^{[2]}=\mu^{[1]}-\Delta_{\mathbf s^{[1]}},
\end{align}
so that $\mu^{[2]}(\mathbf s)=\mu^{[1]}(\mathbf s)=1,\forall \mathbf s\neq \mathbf s^{[1]}$
and $\mu^{[2]}(\mathbf s^{[1]})=0$, hence $\mu^{[2]}$ differs from $\mu^{[1]}$ only in state $\mathbf s^{[1]}$. By leveraging these structural properties into (\ref{problemslope}), we conclude that
\begin{align}
\label{problemslope3}
\mathbf s^{[1]}=\arg\min_{\mathbf s\in\mathcal S}\eta(\mathbf s),
\end{align}
where we have defined the \emph{SU access efficiency} (see also \cite{MichelusiJSAC}) in state $\mathbf s$ as
\begin{align}
\label{etaorig}
\eta(\mathbf s)\triangleq\frac{\bar T_S(\mu^{[1]}-\Delta_{\mathbf s})-\bar T_S(\mu^{[1]})}{\nabla(\mu^{[1]}-\Delta_{\mathbf s})-\nabla(\mu^{[1]})}.
\end{align}
In other words, $\eta(\mathbf s)$ amounts to the \emph{decrease} in SU throughput ($\bar T_S(\mu^{[1]}-\Delta_{\mathbf s})-\bar T_S(\mu^{[1]})$), per unit \emph{decrease} in PU throughput degradation ($\nabla(\mu^{[1]}-\Delta_{\mathbf s})-\nabla(\mu^{[1]})$), as a result of remaining idle in state $\mathbf s$. Since the SU aims at maximizing its own throughput, under a PU throughput degradation constraint, $\mathbf s^{[1]}$ is chosen as the state $\mathbf s$ in (\ref{etaorig}) that minimizes the loss in SU throughput, per unit decrease of the PU throughput degradation, as captured in (\ref{problemslope3}). By solving (\ref{problemslope3}), we obtain the following result.
\begin{lemma}
\label{mu2}
$\mu^{[2]}$ is \emph{uniquely} given by the IC policy
\begin{align}
\label{sdfggfhdsf}
\mu^{[2]}(0)=0,
\qquad
\mu^{[2]}(\mathbf s)=1,\ \forall\mathbf s\in\mathcal S\setminus\{0\}.
\end{align}
\end{lemma}
\begin{proof}
In this proof, we evaluate $\eta(\mathbf s)$ in all states $\mathbf s\in\mathcal S$,
and show that it is minimized by $\mathbf s=0$. We will make use of \arxiv{Appendix A} to compute the performance of $\mu^{[1]}$ and $\mu^{[1]}-\Delta_{\mathbf s},\forall\mathbf s \in\mathcal S$ in closed form, used to compute $\eta(\mathbf s)$ in (\ref{etaorig}).
\iftoggle{arxiv}{}{Due to space constraints, the algebraic steps are provided in \cite{Proofs}.}
 We obtain
\begin{align}
\label{etamum}
&\eta(m)
=
\nonumber
\frac{R_s}{\nabla_{\max}}
\Biggr[
D_s
  - \frac{\zeta(1-\rho_1)(1-\rho_1+\rho_0D_p)}{1-\rho_1(1-D_p)}
  \\&\qquad\qquad
+m\frac{(\rho_1-\rho_0)D_p(1-D_p)}{1-\rho_1(1-D_p)}
\Biggr],
\\
\label{etaKLR}
&\eta(\stackrel{\leftrightarrow}{\K})
=\nonumber
\eta(0)
+\frac{R_s}{\nabla_{\max}}\Biggr[\frac{(\rho_1-\rho_0)(1-D_p)D_p}{1-\rho_1(1-D_p)}
\\&\qquad\qquad
+\frac{[1-\rho_0(1-D_p)](1-\rho_1)\upsilon_s}{[1-\rho_1(1-D_p)]^2}\Biggr],
\\
\label{etaK}
&\eta(\stackrel{\rightarrow}{\K})=
\eta(0)
+\frac{R_s}{\nabla_{\max}}(1-\rho_1)\frac{1-\rho_0(1-D_p)}{1-\rho_1(1-D_p)}\zeta,
\end{align}
where $\eta(0)$ is given by (\ref{etamum}) with $m{=}0$, and $\zeta$ is given by (\ref{zeta}).

To conclude, by comparing the SU access efficiencies, it is clear that $\eta(m){>}\eta(0),\forall m{>}0$,
$\eta(\stackrel{\rightarrow}{\K}){>}\eta(0)$ and $\eta(\stackrel{\rightarrow}{\K}){>}\eta(0)$, so that 
 the solution of (\ref{problemslope}) yields $\mathbf s^{[1]}=0$ and $\mu^{[2]}\equiv\mu^{[1]}-\Delta_{0}$, proving the lemma.
\end{proof}

Under policy $\mu{=}\mu^{[2]}$, we have that $\mu(0){=}0$ and 
$\mu(\stackrel{\rightarrow}{\K}){=}1$ hence, following the discussion in
Sec.~\ref{struct}, the SU accesses only states $0$ (where it remains idle) and $\stackrel{\rightarrow}{\K}$ (where it transmits). Therefore, it transmits with probability one in state $\stackrel{\rightarrow}{\K}$ only, \emph{i.e.}, after the PU packet becomes known at SUrx; when this happens, the SU packet is decoded via interference cancellation. For this reason, $\mu^{[2]}$ is termed "IC policy." As discussed in Sec.~\ref{struct}, under such policy we obtain $\nabla(\mu^{[2]}){=}\nabla_{GA}$ as in (\ref{NAGA}). Thus, if $\nabla_{th}{\in}[\nabla(\mu^{[2]}){=}\nabla_{GA},\nabla(\mu^{[1]}){=}\nabla_{\max}]$, the optimal policy is obtained by randomizing between the IC policy $\mu^{[2]}$ and the Always-TX policy $\mu^{[1]}$, or equivalently, by the policy (\ref{pol2}) randomized in state $0$. The optimal policy given in (\ref{pol2}) and its performance in (\ref{perf2}) are obtained by enforcing $\nabla(\mu^*){=}\nabla_{th}$ to determine the optimal value of $\mu^*(0)$. On the other hand, if $\nabla_{th}{\geq}\nabla_{\max}$, then the optimal policy is Always-TX ($\mu^{[1]}$), which maximizes the SU throughput and satisfies the constraint $\nabla(\mu^{[1]}){=}\nabla_{\max}{\leq}\nabla_{th}$. Thus, we have proved the structure and performance of the optimal policy for the case $\nabla_{th}{\geq}\nabla_{GA}$ as well.
\vone{-5mm}
\vdo{-3mm}
\section{Online Learning and Adaptation}
\label{learning}
The Always-TX, IC and Idle modes do not require any knowledge of the statistics of the model, such as the decoding probabilities (\ref{outcomes}) or the PU outage probabilities $\rho_0$ and $\rho_1$. Thus, the SU needs only to learn the optimal randomization among these three modes of operation. This can be inferred from the throughput degradation experienced by the PU, estimated by monitoring the ACK/NACK feedback: when this estimate is below the PU throughput degradation constraint, the SU may transmit more often by favoring the Always-TX or IC mode (depending on the time-sharing currently in use); when this estimate is above the constraint, the SU may reduce its transmissions by favoring the IC or Idle modes. This feature of the optimal policy facilitates learning and adaptation in practical settings where the statistics of the system are unknown, or vary over time.

In this section, we propose an algorithm based on \emph{stochastic gradient descent} (SGD) \cite[Chapter 14]{Shalev-Shwartz} for the online optimization of the SU access policy and transmit rate $R_s$, by leveraging the structure of the optimal access policy. Note that $\mu^*$ in Theorems~\ref{thm:policy1} and~\ref{thm:policy2} is uniquely characterized by a parameter $\nu{\triangleq}\mu^*(0){+}\mu^*(\stackrel{\rightarrow}{\K})\in[0,2]$, related to the \emph{access level} of the SU: given $\nu$, we get $\mu^*$ as $\mu^*(0){=}\max\{\nu-1,0\}$ and $\mu^*(\stackrel{\rightarrow}{\K}){=}\min\{\nu,1\}$, and the degradation to the PU, $\nabla_{th}$, is related to $\nu$ via (\ref{pol1}) for $0{\leq}\nu{\leq}1$ and (\ref{pol2}) for $1{<}\nu{\leq}2$. Let $\bar T_S(\nu)$ and $\bar T_P(\nu)$ be the corresponding SU and PU throughputs, which are increasing and decreasing functions of $\nu$, respectively. Let $\bar T_{P,\min}$ be the minimum throughput requirement for the PU. This information may be broadcast by the PU system to regulate the access of SUs. The PU throughput degradation constraint $\nabla_{th}$ in (\ref{optP1}) is related to $\bar T_{P,\min}$ via $\nabla_{th}=1-\bar T_{P,\min}/\bar T_{P,\max}$. The rate $R_s$ is chosen so as to maximize the SU throughput under no interference from the PU signal, \emph{i.e.}, $R_s=\arg\max_{r_s} r_s\mathbb P(r_s{<}C\left(\gamma_s\right))$. Under Rayleigh fading, we obtain $\mathbb P(r_s{<}C\left(\gamma_s\right)){=}\exp\{-(2^{r_s}-1)/\bar\gamma_s\}$.
Thus, the optimal $\nu^*{\in}[0,2]$ (or equivalently, the optimal policy $\mu^*$) and $R_s^*{\geq}0$ can be expressed as the minimizers of
\begin{align}
\label{optnuRs}
\min_{\nu,r_s} \frac{1}{2}(\bar T_P(\nu)-\bar T_{P,\min})^2-
  r_s\exp\left\{-\frac{2^{r_s}-1}{\bar\gamma_s}\right\}.
\end{align}
 We denote the objective function as $G(\nu,r_s)$. Consider the optimization with respect to $\nu$. Since $\bar T_P(\nu)$ is a decreasing function of $\nu$, if $\bar T_{P,\min}>\bar T_{P,\max}$ (hence, $\bar T_P(\nu)<\bar T_{P,\min}$), then the solution is $\nu=0$ (the SU remains idle, and the optimization of $R_s$ is irrelevant); indeed, in this case, the PU has set an unrealistic demand, hence the SU should remain idle to at least partially satisfy it. If $R_p(1{-}\rho_{1})\leq\bar T_{P,\min}\leq \bar T_{P,\max}$, where $R_p(1{-}\rho_{1})$ is the PU throughput achieved when the SU always transmits, then the solution is the unique $\nu^*$ such that $\bar T_P(\nu^*)=\bar T_{P,\min}$, \emph{i.e.}, the PU throughput constraint
is attained with equality. Finally, if $\bar T_{P,\min}<R_p(1{-}\rho_{1})$, then $\bar T_{P,\min}<\bar T_P(\nu),\forall \nu$, hence the solution is $\nu=2$ (the SU always transmits); indeed, in this case, the PU demand can be met even if the SU always transmits. 

Problem (\ref{optnuRs}) can be solved using the \emph{gradient projection algorithm} \cite[Chapter 3]{Bertsekas_parallel}. The gradient of the objective function $G(\nu,r_s)$ with respect to $\nu$ and $r_s$ is given by
\begin{align}
\nonumber
&\frac{\mathrm d G(\nu,r_s)}{\mathrm d\nu}
=
\frac{\mathrm d \bar T_P(\nu)}{\mathrm d\nu}
(\bar T_P(\nu)-\bar T_{P,\min})
\\&
\propto\bar T_{P,\min}-\bar T_P(\nu)
\triangleq g_1(\nu),
\\
\nonumber
&\frac{\mathrm d G(\nu,r_s)}{\mathrm d r_s}
=
\exp\left\{-\frac{2^{r_s}-1}{\bar\gamma_s}\right\}
\left[
\frac{1}{\bar\gamma_s}\ln(2)r_s2^{r_s}-1
\right]
\\&
\propto
\mathbb E[a_{S,t}]\left[\ln(2)r_s2^{r_s}-\bar\gamma_s\right]\triangleq g_2(r_s),
\end{align}
where $\propto$ denotes proportionality up to a positive multiplicative factor, since $\mathrm d\bar T_P(\nu)/\mathrm d\nu{<}0$. Thus, (\ref{optnuRs}) can be solved as
\begin{align}
&\nu_{t+1}=\left[\nu_t-\beta_t g_1(\nu_t)\right]_0^2,
\\
&R_{s,t+1}=\left[R_{s,t}-\beta_t g_2(R_{s,t})\right]^+,
\end{align}
where $\beta_t>0$ is the step-size, $[\cdot]_0^2=\min\{\max\{\cdot,0\},2\}$ and $[\cdot]^+=\max\{\cdot,0\}$ are projection operations onto the feasible sets. The policy used at time $t$ is then given by
\begin{align}
\label{polti}
\left\{\begin{array}{l}
\mu_t(0)=\max\{\nu_t-1,0\},\\
\mu_t(\stackrel{\rightarrow}{\K})=\min\{\nu_t,1\},\\
\mu_t(\stackrel{\leftrightarrow}{\K})=\mu_t(b)=1,\ \forall b>0.
\end{array}\right.
\end{align}
However, typically $\bar T_P(\nu)$ may not be available to the SU to compute the gradient $g_1(\nu)$, but only observations of the ACK/NACK feedback sequence $\{y_{P,t},t\geq 0\}$;
similarly, only realizations of the channel fading $\gamma_{s,t}$ may be available via channel estimation, instead of the expected channel gain $\bar\gamma_s=\mathbb E[\gamma_{s,t}]$ required to compute $g_2(r_s)$. Thus, we use the SGD algorithm,
which replaces  $g_1(\nu_t)$ and $g_2(r_{s,t})$ with estimates $\hat g_{1,t}$ and $\hat g_{2,t}$ such that $\mathbb E[\hat g_{1,t}|\nu_t]=g_1(\nu_t)$ and $\mathbb E[\hat g_{2,t}|\nu_t]=g_2(R_{s,t})$. In particular, we choose
\begin{align}
&\hat g_{1,t}=\bar T_{P,\min}-R_p\chi(y_{P,t}=\text{ACK}),\\
&\hat g_{2,t}=a_{S,t}\left[\ln(2)R_{s,t}2^{R_{s,t}}-\gamma_{s,t}\right].
\end{align}
We finally obtain
\begin{align}
\label{SGDalgo}
&\nu_{t+1}=\left[\nu_t+\beta_t(R_p\chi(y_{P,t}=\text{ACK})-\bar T_{P,\min})\right]_0^2,
\\
&R_{s,t+1}=\left[R_{s,t}+\beta_t a_{S,t}\left[\gamma_{s,t}-\ln(2)R_{s,t}2^{R_{s,t}}\right]\right]^+,
\label{Rst}
\end{align}
where $\nu_0=0$ and $R_{s,0}=0$ (the SU is idle at initialization). Thus, $\nu$ tends to augment if an ACK is received, so that the SU may transmit more often, and to diminish otherwise; $R_s$ tends to augment if the channel is good ($\gamma_{s,t}>\ln(2)R_{s,t}2^{R_{s,t}}$), and diminish otherwise. In static scenarios where the parameters of the model do not change, a decreasing step-size is commonly used in stochastic optimization, such as
$\beta_t=\beta_0/(t+1)$; in time-varying scenarios, a fixed but small step-size may be used, in order to accommodate adaptation. Note that, if $a_{S,t}=0$, the SU remains idle and the channel fading may not be estimated, yielding $R_{s,t+1}=R_{s,t}$ as in (\ref{Rst}).
\vone{-3mm}
\vdo{-4mm}
\section{Numerical Results}
\label{sec:numres}
In this section, we present numerical results. PUtx is located at position $(0,0)$, PUrx at $(0,d_0)$, at reference distance $d_0$ from PUtx. SUtx and SUrx are located at positions
$(d_{SP},0)$ and $(d_{SP},d_0)$, respectively, where $d_{SP}$ is the distance between the SU and PU pairs. We assume Rayleigh fading channels. The expected SNR of the link PUtx-PUrx is $\bar\gamma_p{=}20$. For any other link, the expected SNR is given by $\bar\gamma_{TR}{=}\bar\gamma_p\left(d_{TR}/d_0\right)^{-\alpha}$, where $d_{TR}$ is the distance between the corresponding transmitter (T) and receiver (R), and $\alpha{=}2$ is the pathloss exponent. $R_p$ and $R_s$  are chosen so as to maximize the respective PU and SU throughputs under no interference, \emph{i.e.}, $R_x=\arg\max_{r_x}r_x\mathbb P(r_x{<}C\left(\gamma_x\right)),x\in\{p,s\}$. The outage probabilities for the PU are computed as $\rho_0=\mathbb P(R_s{<}C(\gamma_{p}))$ and $\rho_1{=}\mathbb P(R_s{<}C(\gamma_{p}/(1{+}\gamma_{sp})))$.

\begin{figure}[t]
\centering  
\includegraphics[width=\w{.75}{.5}\linewidth,trim = 5mm 0mm 18mm 10mm,clip=false]{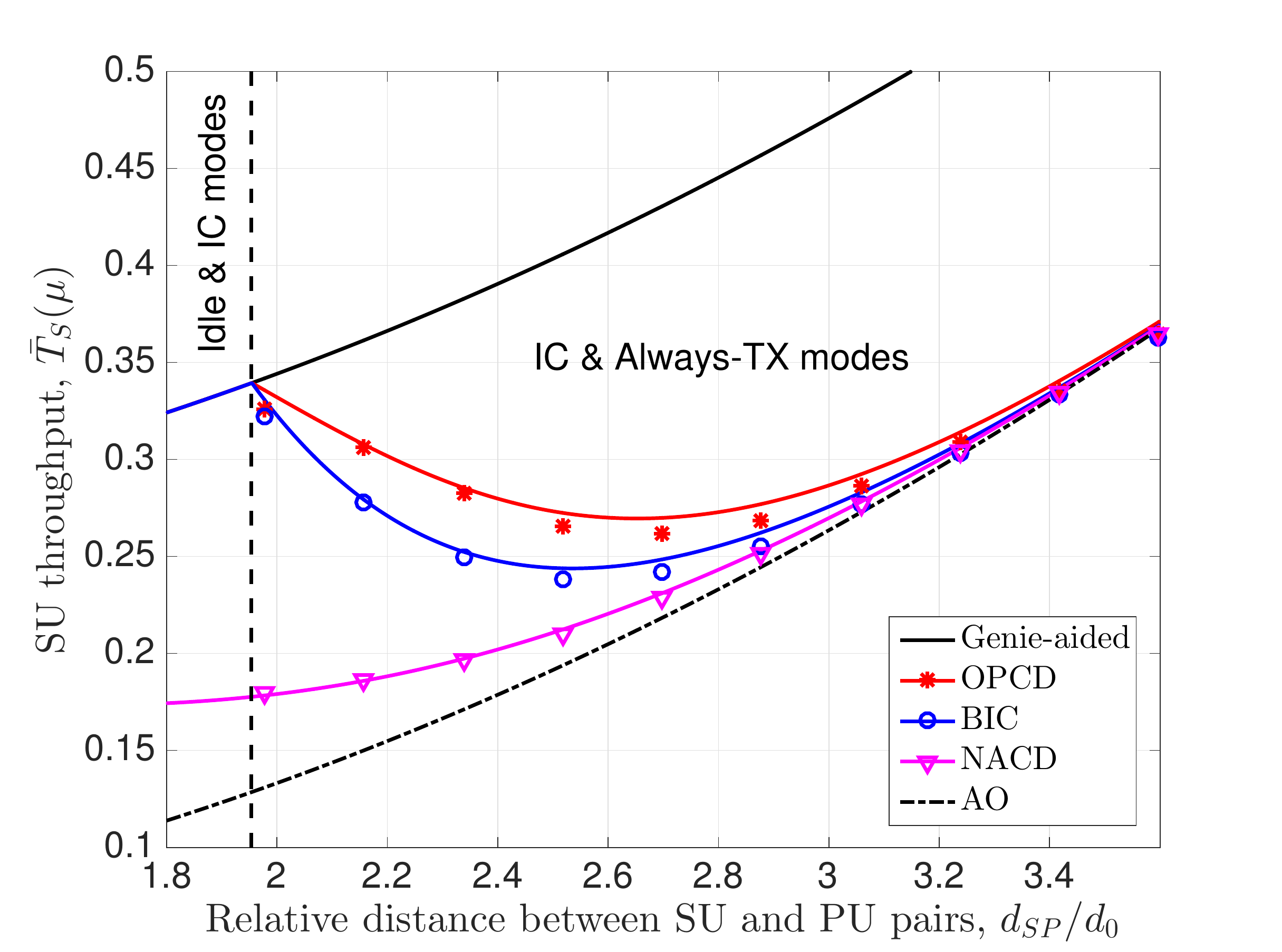}
\caption{SU throughput versus relative distance between SU and PU pairs, $d_{SP}/d_0$;
$\nabla_{th}=0.1$.
Solid lines: analytical expression; Markers: SGD algorithm via Monte Carlo simulation.}
\label{results}
\vone{-7mm}
\vdo{-5mm}
\end{figure}

We consider the following schemes: 1) The optimal CD policy (OPCD) given by Theorems~\ref{thm:policy1} and~\ref{thm:policy2}. 2) The BIC scheme developed in  \cite{MichelusiJSAC}, where, unlike CD, the SU does not perform retransmissions; hence,
after decoding the PU packet, it uses interference cancellation only within the current ARQ window (see example in Fig.~\ref{figexlabel2}). 3) The non-adaptive CD scheme (NACD), where the SU adopts the optimal packet selection policy without any access policy, \emph{i.e.}, it transmits with constant probability $\min\{\nabla_{th}/\nabla_{\max},1\}$ in all slots, independently of the state. 4) The "ARQ-oblivious" (AO) scheme, originally proposed in \cite{IT_ARQ}, where the SU attempts to jointly decode the SU and PU packets and remove the interference of the latter, by leveraging the PU codebook structure \cite{Taranto}; however, it does not exploit the redundancy  of the ARQ mechanism to perform interference cancellation over the ARQ window; its performance is given by $\bar T_{AO}=\min\{\nabla_{th}/\nabla_{\max},1\}R_s(\delta_s+\delta_{sp})$. Additionally, we plot the genie-aided throughput $\bar T_S^{(GA)}(\min\{\nabla_{th}/\nabla_{\max},1\})$, which assumes a-priori knowledge of the PU packet.

In Fig.~\ref{results}, we plot the SU throughput as a function of $d_{SP}/d_0$. The interference constraint to the PU is set to $\nabla_{th}=0.1$, so that the SU is allowed to degrade the PU throughput by at most $10\%$. The solid lines refer to the analytical expressions, whereas the markers refer to the SGD algorithm developed in Sec.~\ref{learning}, evaluated via Monte Carlo simulation over  $10^5$ slots. The SGD algorithm
is initialized with $\nu_{0}=0$ and $R_{s,0}=0$, so that the SU is initially idle. This is a conservative behavior, which minimizes the risk of generating harmful interference to the PU in the initial phase when the SU is uninformed, and allows the latter to collect observations before undertaking a more active behavior. We note that SGD closely approaches the analytical curves.

Note that AO lower bounds the performance of CD and BIC, since it does not leverage the interference of the ARQ protocol. In contrast, by assuming non-causal knowledge of the PU packet, "genie-aided" upper bounds the performance. Both AO and "genie-aided" exhibit a monotonically increasing trend as a function of $d_{SP}$: as the SU and PU pairs move farther away from each other, the interference decreases, hence the SU can afford to transmit more frequently. Similarly, the SU throughputs of OPCD and BIC monotonically increase and attain the genie-aided throughput for $d_{SP}\lessapprox 2d_0$. In fact, in this regime, SUtx is close to PUrx, hence its transmissions interfere strongly with PUrx; in this case, the SU may only transmit sparingly by randomizing between Idle and IC modes (policy (\ref{pol1})), which attains the genie-aided upper bound, see (\ref{GAth}). As $d_{SP}$ increases beyond $2d_0$, the interference to the PU becomes weaker, the SU may transmit more frequently, hence it starts transmitting in state $0$ as well (policy (\ref{pol2})), creating a gap with respect to the genie-aided throughput (case $\nabla_{GA}{<}\nabla_{th}{<}\nabla_{\max}$ in Theorem~\ref{thm:policy2}).

\begin{figure}[t]
\centering  
\includegraphics[width=\w{.75}{.5}\linewidth,trim = 5mm 0mm 18mm 10mm,clip=false]{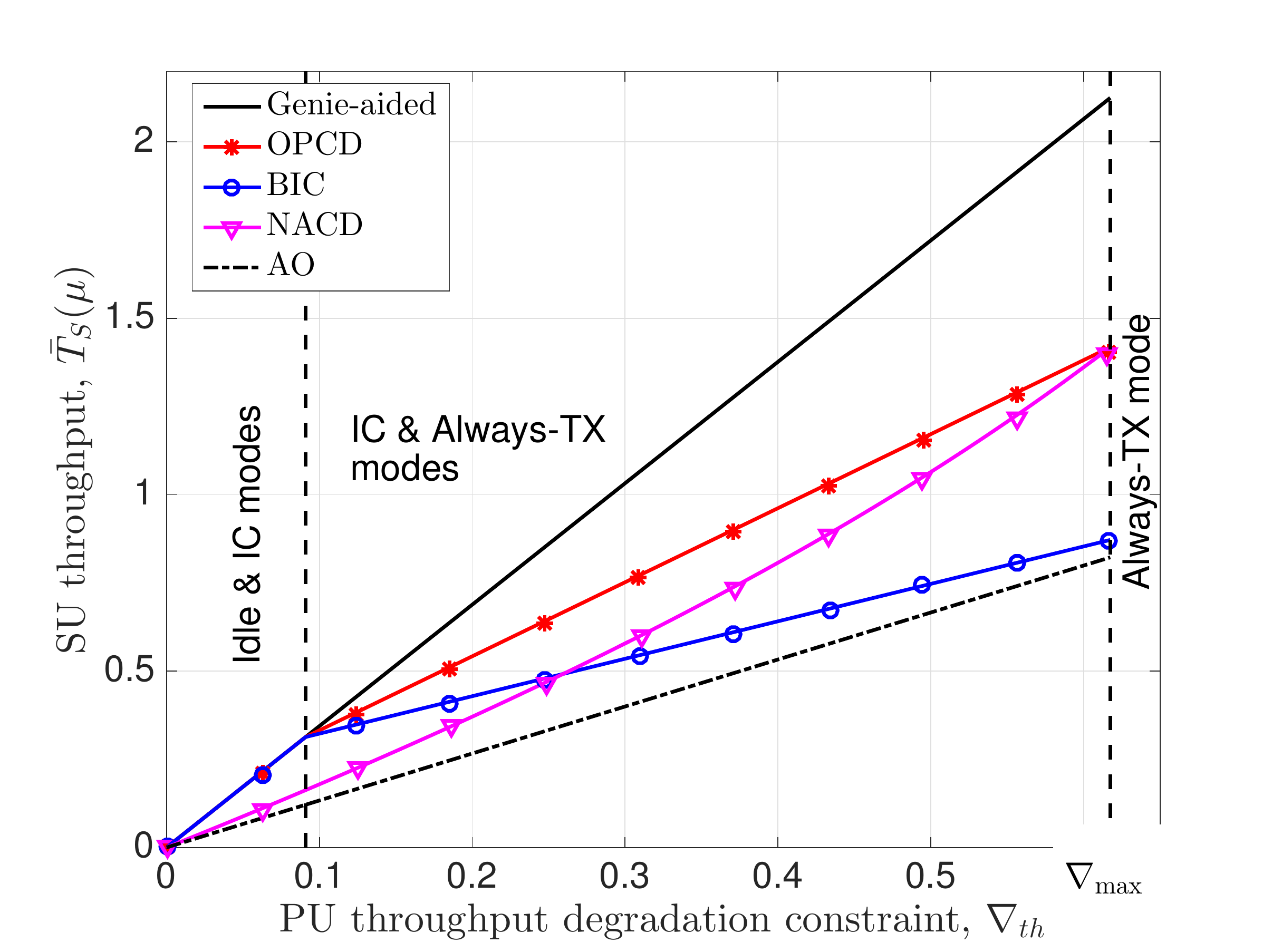}
\caption{SU throughput versus $\nabla_{th}$; $d_{SP}=2 d_0$.
Solid lines: analytical expression; Markers: SGD via Monte Carlo simulation.}
\label{results2}
\vone{-7mm}
\vdo{-5mm}
\end{figure}

However, surprisingly, OPCD and BIC do not follow the same monotonic trend as AO and "genie-aided" for $d_{SP}\gtrapprox 2d_0$: the performance of OPCD decreases for $2d_0\lessapprox d_{SP}\lessapprox 2.7d_0$, and approaches that of AO for $d_{SP}{\to}3d_0$ (a similar consideration holds for BIC in the range $2d_0\lessapprox d_{SP}\lessapprox 2.5d_0$). This counterintuitive result can be explained as follows:
\begin{enumerate}[leftmargin=0.45cm]
\item As  SUrx moves farther away from PUtx, SUrx receives a weaker PU signal, hence it becomes more difficult to decode the PU packet and remove its interference, and, in turn, to decode the CD root and initiate CD;
\item  As  SUtx moves farther away from PUrx, the interference generated by SUtx to PUrx becomes weaker, hence the outage probability $\rho_1$ tends to decrease; therefore, when using the Always-TX mode, the ARQ window tends to shorten (its average duration is $1/(1-\rho_1)$ in the  Always-TX mode) since the PU is more likely to succeed, resulting in fewer opportunities to leverage the redundancy of the ARQ protocol.
\end{enumerate}
Overall, we note that OPCD outperforms BIC by up to $\sim 15\%$ and achieves up to $2\times$ throughput improvement over NACD and up to $3\times$ over AO. NACD performs poorly compared to both OPCD and BIC, revealing the importance of using an optimized SU access policy over a non-adaptive one which does not fully leverage the structure of the problem.

\begin{figure}[t]
\centering  
\includegraphics[width=\w{.75}{.5}\linewidth,trim = 5mm 0mm 18mm 10mm,clip=false]{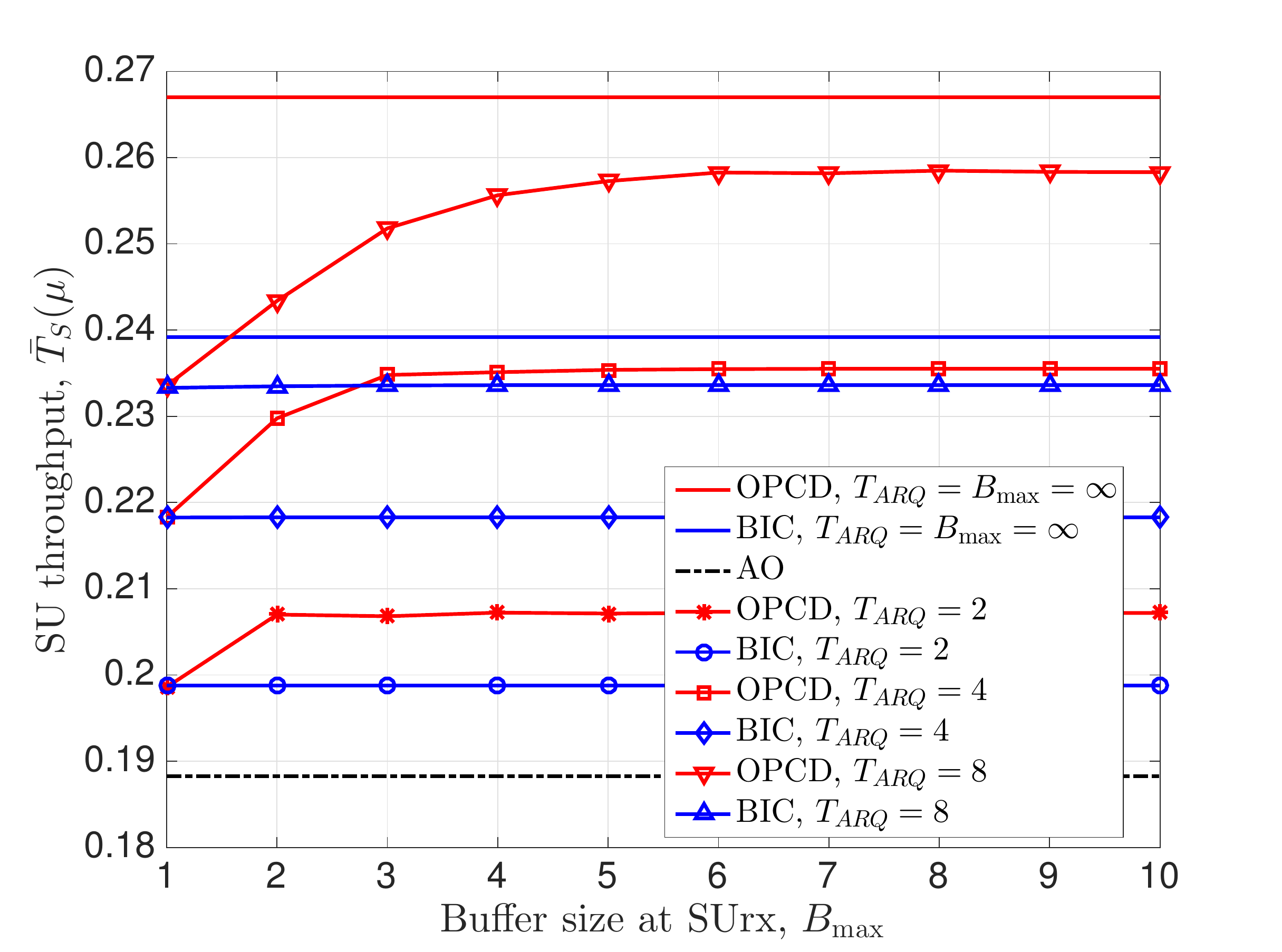}
\caption{SU throughput versus SUrx buffer size $B_{\max}$ and PU ARQ deadline $T_{ARQ}$; $d_{SP}=2.5 d_0$, $\nabla_{th}=0.1$.}
\label{results3}
\vone{-6mm}
\vdo{-2mm}
\end{figure}

\begin{figure}[t]
\centering  
\includegraphics[width=\w{.75}{.5}\linewidth,trim = 5mm 0mm 18mm 10mm,clip=false]{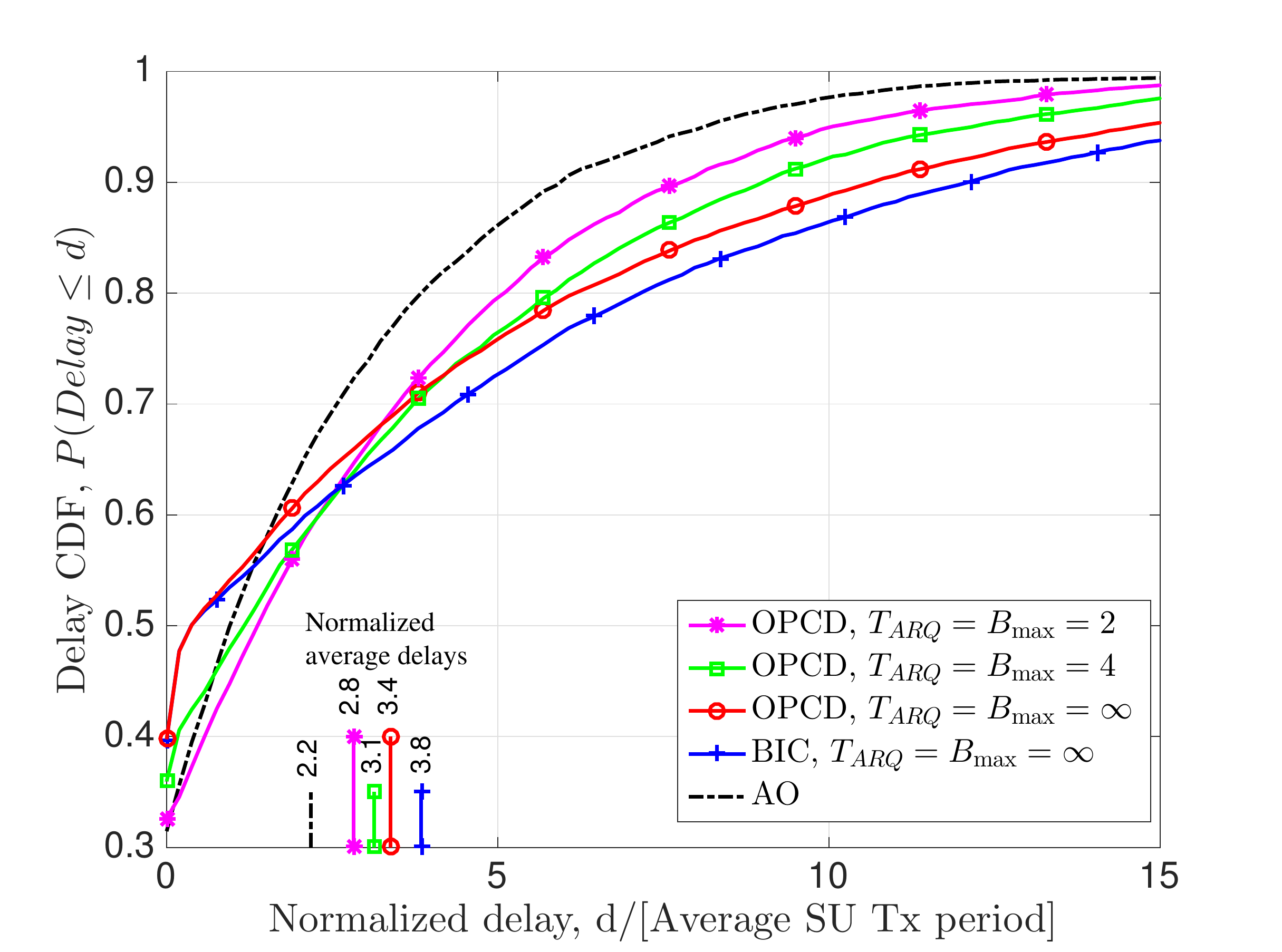}
\caption{Cumulative distribution function (CDF) of the SU delay,
 normalized to the average transmission period of the SU ($\simeq 5.3$ [slots] in this case);
$d_{SP}=2.5 d_0$, $\nabla_{th}=0.1$.}
\label{results4}
\vone{-3mm}
\vdo{-6mm}
\end{figure}

In Fig.~\ref{results2}, we plot the trade-off between the SU throughput $\bar T_S(\mu)$ and the PU throughput degradation $\nabla(\mu)$ as $\nabla_{th}$ is varied. For $\nabla_{th}\leq\nabla_{GA}\simeq 0.1$, OPCD randomizes between the Idle and IC modes; for $\nabla_{GA}\leq\nabla_{th}\leq\nabla_{\max}\simeq 0.6$, it randomizes between the IC and Always-TX modes; for $\nabla_{th}\geq\nabla_{\max}\simeq 0.6$ the interference constraint becomes inactive and Always-TX is the only mode of operation. In all cases, $\bar T_S(\mu)$ monotonically increases with $\nabla_{th}$, since more opportunities become available to the SU to use the channel. OPCD achieves up to 40\%-60\% improvement over BIC, for $\nabla_{th}\gtrapprox 0.3$, up to $2\times$ improvement over NACD for $\nabla_{th}\lessapprox 0.1$, and $\geq 70\%$ improvement over AO, for all range of values.

\begin{figure}[t]
\centering  
\includegraphics[width=\w{.75}{.5}\linewidth,trim = 5mm 0mm 18mm 10mm,clip=false]{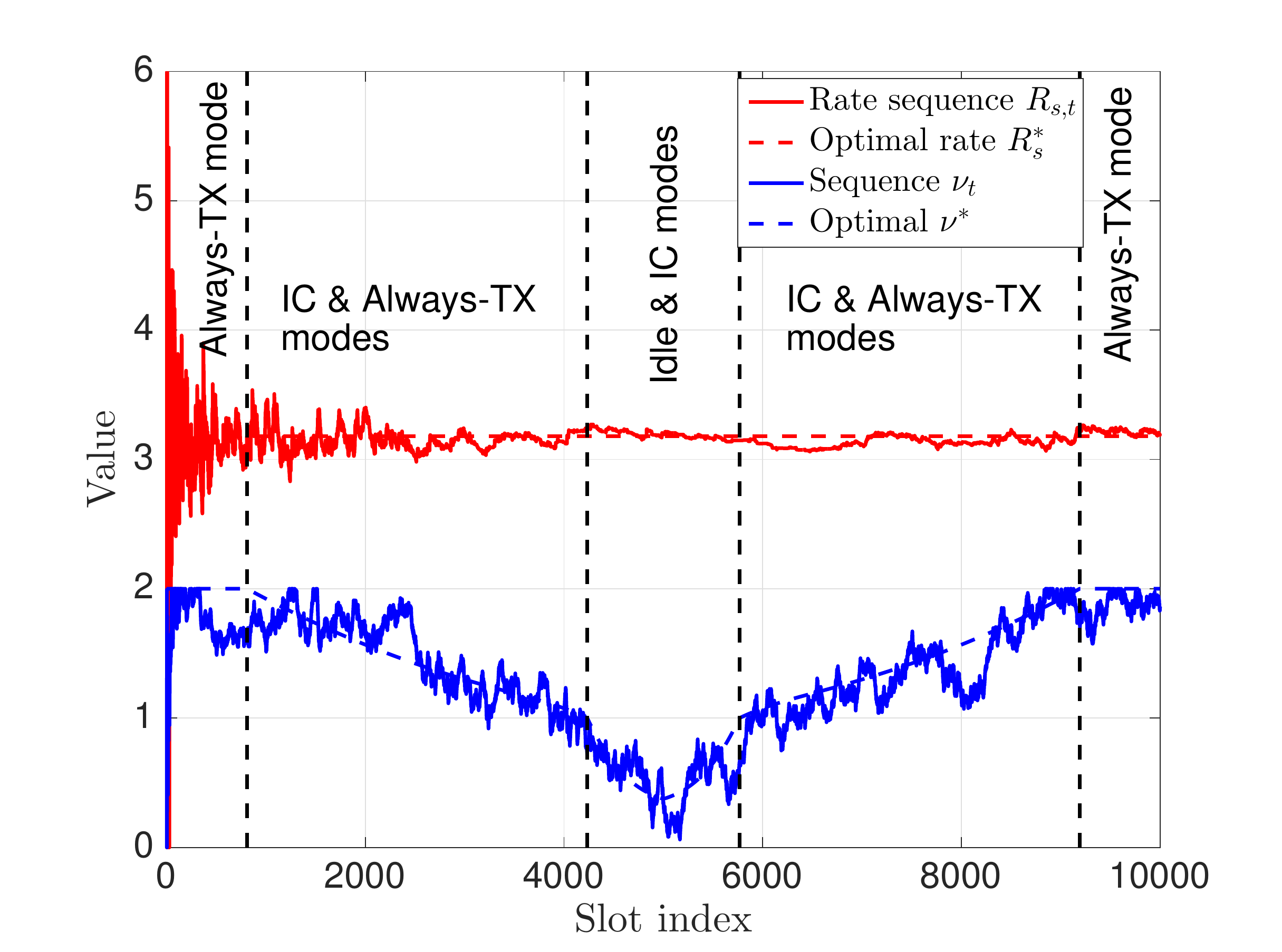}
\caption{Convergence of the SGD algorithm; $\nabla_{th}{=}0.1$;
the distance varies linearly from $d_{SP}{=}10$ at time $0$ to $d_{SP}{=}0.5$ at time $t{=}5000$, and then linearly to $d_{SP}{=}10$ at time $t{=}10^4$.}
\label{results5}
\vone{-9mm}
\vdo{-6mm}
\end{figure}

In the analysis, it was assumed that the PU packet is transmitted until successfully decoded at its intended receiver and an infinite buffer to store the received signals at SUrx. In practical systems, these are finite quantities. In Fig.~\ref{results3}, we evaluate the effect of a finite
ARQ deadline $T_{ARQ}$ and a finite buffer size $B_{\max}$. Each unit corresponds to 
the buffer space required to store one received signal. When the buffer at SUrx is full, the received signal is discarded after the decoding attempt, resulting in missed opportunities to build up the CD graph. Note that, as $T_{ARQ}$ increases, the performance improves. In fact, the longer the ARQ window, the more opportunities  available at the SU to leverage the redundancy of the ARQ process of the PU. Surprisingly, most of the benefits of CD are reaped with a buffer size of only $B_{\max}{\simeq}4$ units, with larger buffer sizes yielding only marginal improvements. Generally, OPCD outperforms BIC, demonstrating a better use of the buffer space available at SUrx.

The CD protocol may introduce delay at the SU, due to the buffering mechanism at SUrx. In Fig.~\ref{results4}, we investigate the delay CDF. We assume that a higher layer SU protocol manages retransmissions of failed attempts (in the case of OPCD or BIC, the retransmission of SU packets that cannot be buffered at SUrx). Since, on average, the SU transmits with probability $\mu_{avg}=\frac{1-\rho_{0}}{\rho_{1}-\rho_{0}}\nabla_{th}$ to obey the maximum PU throughput degradation constraint, we normalize the delay to the average transmission period of the SU, $1/\mu_{avg}\simeq 5.3$ [slots]. Remarkably, OPCD outperforms BIC, both in terms of delay and throughput (Fig.~\ref{results3}), thanks to a more efficient use of the buffer at SUrx. As expected, AO outperforms both OPCD and BIC (while delivering the worst throughput, see Fig.~\ref{results}), since it does not use a buffering mechanism at SUrx, but simply retransmits packets in case of failure. The only exception is the normalized delay region $[0,1.3]$, where both OPCD and BIC outperform AO (case $B_{\max}=T_{ARQ}=\infty$). In fact, under OPCD or BIC, SUrx leverages knowledge of the PU packet to perform interference cancellation, and thus fewer attempts are needed to succeed in data transmission. Finally, we notice that the average delay of OPCD increases with $B_{\max}$, as a result of an increased buffering capability at SUrx. In general, we observe a trade-off between throughput (Fig.~\ref{results3}) and delay (Fig.~\ref{results4}) by varying $B_{\max}$.

In Fig.~\ref{results5}, we investigate the performance of the SGD algorithm developed in Sec.~\ref{learning} in a time-varying scenario with constant step-size $\beta_t$. The distance between SUtx and SUrx is kept fixed, whereas that between the SU and PU pair, $d_{SP}$, varies as described in the caption. Accordingly, the optimal value of the SU rate $R_s$ is constant, whereas the optimal value of the parameter $\nu$ depends on $d_{SP}$: when $d_{SP}$ is large (around $t\simeq 0$ and $t\simeq 10000$, the SU and PU pairs are far away from each other), $\nu^*$ is large since the SU generates little interference to the PU and thus may transmit more frequently; in contrast, when $d_{SP}$ is small (around $t\simeq 5000$, the SU and PU pairs are close to each other), $\nu^*$ is small since the SU generates strong interference to the PU, hence it may only transmit sparingly. We note that the SGD algorithm, after an initial convergence phase, closely tracks the optimal values of $R_s^*$ and $\nu^*$. Surprisingly, this is accomplished by only observing the ARQ feedback from PUrx and the channel fading realization, see (\ref{SGDalgo}).

\vone{-5mm}
\vdo{-3mm}
\section{Conclusions}\label{sec:conclu}
In this paper, we have investigated the design of optimal SU access policies that maximize the SU throughput via chain decoding, subject to an interference constraint to the PU. We have found a closed form expression of the optimal policy and of its performance, and shown that it can be expressed as a randomization among three modes of operation. We have designed an algorithm based on stochastic gradient descent to determine the optimal randomization in practical settings where the statistics of the system are unknown or vary over time. We have shown numerically that, for a 10\% interference constraint, the optimal access policy with chain decoding outperforms by 15\% a state-of-the-art scheme that does not exploit opportunistic \emph{secondary} retransmissions, and achieves up to $2\times$ improvement over a chain decoding scheme using a non-adaptive access policy instead of the optimal one.
 
\vone{-5mm}
\vdo{-3mm}
\bibliographystyle{IEEEtran}
\bibliography{IEEEabrv,References} 

\iftoggle{doublecol}{
\vdo{-5mm}
\begin{IEEEbiography}
    [{\includegraphics[width=1in,height=1.25in,clip,keepaspectratio]{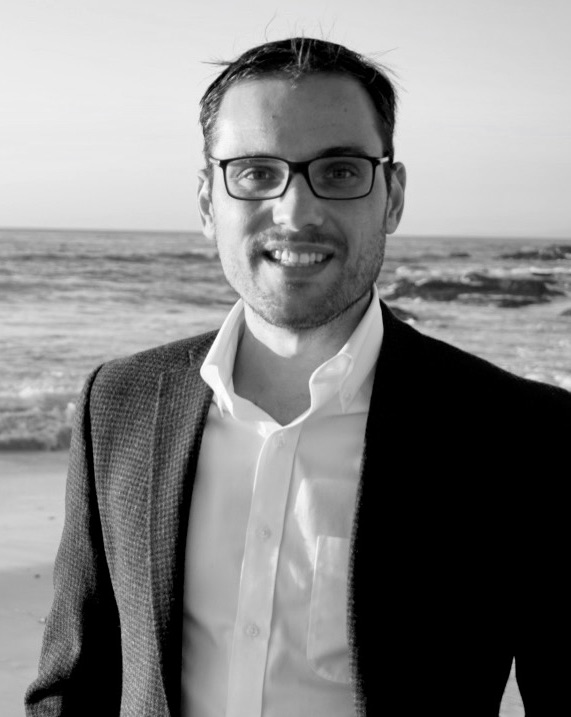}}]{Nicolo Michelusi}(S'09, M'13, SM'18) received the B.Sc. (with honors), M.Sc. (with honors) and Ph.D. degrees from the University of Padova, Italy, in 2006, 2009 and 2013, respectively, and the M.Sc. degree in Telecommunications Engineering from the Technical University of Denmark in 2009, as part of the T.I.M.E. double degree program. He was a post-doctoral research fellow at the Ming-Hsieh Department of Electrical Engineering, University of Southern California, USA, in 2013-2015. He is currently an Assistant Professor at the School of Electrical and Computer Engineering at Purdue University, IN, USA. His research interests lie in the areas of 5G wireless networks, millimeter-wave communications, stochastic optimization, distributed optimization. Dr. Michelusi serves as Associate Editor for the IEEE Transactions on Wireless Communications, and as a reviewer for several IEEE Transactions.
\end{IEEEbiography}}

\end{document}